\documentclass[11pt,a4paper]{article}
\usepackage{amsfonts}
\usepackage{graphicx}
\usepackage{indentfirst}
\usepackage{amsmath,amssymb,amsthm}
\usepackage{amssymb}
\usepackage{algorithm,algpseudocode}
\usepackage{latexsym}
\usepackage{natbib}
\usepackage[margin=1in,a4paper]{geometry}
\usepackage{colortbl}
\usepackage{color}
\usepackage[colorlinks,citecolor=blue,urlcolor=blue]{hyperref}
\usepackage{mathtools}
\usepackage{setspace}
\usepackage{verbatim}
\usepackage{subfigure}
\usepackage[title]{appendix}
\usepackage{bm}
\usepackage{enumitem}

\mathtoolsset{showonlyrefs=true}

\newtheorem{proposition}{Proposition}
\newtheorem{corollary}{Corollary}
\newtheorem{lemma}{Lemma}

\newtheorem{assumption}{Assumption}
\newtheorem{theorem}{Theorem}
\theoremstyle{remark}
\newtheorem{remark}{Remark}

\allowdisplaybreaks[4]

\graphicspath{{../../Codes/}}

\newcommand*\diff{\mathop{}\!\mathrm{d}}

\title{Drawdowns, Drawups, and Occupation Times under General Markov Models}

\author{Pingping Zeng\thanks{Department of Mathematics, Southern University of Science and Technology, Shenzhen, China. Email: zengpp@sustech.edu.cn.} \and Gongqiu Zhang\thanks{Corresponding author. School of Science and Engineering, The Chinese University of Hong Kong (Shenzhen), China. Email: zhanggongqiu@cuhk.edu.cn.} \and Weinan Zhang\thanks{Department of Data Science, City University of Hong Kong, Hong Kong SAR. Email: zhang.weinan@foxmail.com.}}

\begin{document}

\maketitle

\begin{abstract}
 Drawdown risk, an important metric in financial risk management, poses significant computational challenges due to its highly path-dependent nature. This paper proposes a unified framework for computing five important drawdown quantities introduced in \cite{landriault2015frequency-drawdown} and \cite{zhang2015occupation-drawdown} under general Markov models. We first establish linear systems and develop efficient algorithms for such problems under continuous-time Markov chains (CTMCs), and then establish their theoretical convergence to target quantities under general Markov models. Notably, the proposed algorithms for most quantities achieve the same complexity order as those for path-independent problems: cubic in the number of CTMC states for general Markov models and linear when applied to diffusion models.
 Rigorous convergence analysis is conducted under weak regularity conditions, and extensive numerical experiments validate the accuracy and efficiency of the proposed algorithms.
\end{abstract}



\section{Introduction}
\label{sec:introduction}

Drawdown, defined as the magnitude of decline of the asset price from its historical peak, has emerged as a widely adopted downside risk metric in financial markets. Its prominence is evidenced by growing scholarly and industrial interest, particularly in applications such as performance evaluation of investment funds and algorithmic trading strategies, where maximum drawdown is routinely reported as a key robustness indicator. Notably, empirical studies increasingly associate irrational trading behaviors, including panic selling and momentum overreaction, with exposure to extreme drawdown events. 

Various drawdown quantities have been systematically examined in the academic literature, including the first passage time of drawdown processes 
\citep{taylor1975stopped,lehoczky1977formulas,douady2000probability,magdon2004maximum,hadjiliadis2006drawdowns,pospisil2009formulas,zhang2010drawdowns,mijatovic2012drawdown,zhang2015occupation-drawdown,landriault2017unified,zhang2023drawdown}, maximum drawdown \citep{magdon2004maximum-risk,pospisil2008pde,pospisil2010portfolio,schuhmacher2011sufficient}, speed and duration of drawdown \citep{zhang2012drawdowns,landriault2017magnitude,cui2018magnitude,li2024speed-duration-drawdown}, along with drawdown-linked insurance products and derivatives \citep{shepp1993russian,asmussen2004russian,avram2004exit,vecer2006option,vecer2007preventing,zhang2013stochastic,landriault2015frequency-drawdown,cui2016omega,zhang2021american-drawdown}. Furthermore, drawdown constraints have been rigorously embedded within stochastic control frameworks for portfolio optimization \citep{grossman1993optimal,cvitanic1994portfolio,chekhlov2005drawdown,cherny2013portfolio}, marking a paradigm shift in risk-aware asset allocation methodologies.

Among the extensive literature on drawdown and related derivatives, seminal works by \cite{landriault2015frequency-drawdown} and \cite{zhang2015occupation-drawdown} made significant strides in quantifying the frequency of drawdowns, drawdowns preceding drawups (vice versa), and occupation times of drawdown processes. These drawdown quantities are especially useful in financial engineering applications and of theoretical importance in stochastic modeling. However, these foundational studies predominantly assume a diffusion process or even a Brownian motion, as the underlying model, fail to capture the empirical reality of financial markets where discontinuous price movements (jumps) frequently occur. To address this limitation and enable robust evaluation of these drawdown quantities under general Markov models, we develop a novel computational framework based on the  continuous-time Markov chain (CTMC) approximation.

For a stochastic process $\{X_t\}_{t\geq 0}$, the drawdown process $\{D_t\}_{t\geq 0}$ and drawup process $\{\widehat D_t\}_{t\geq 0}$ can be formally defined as
\begin{equation*}
    \label{eq:drawdown_drawup}
    D_t = \sup_{0 \le s\leq t} X_s - X_t~~\text{and}~~ \widehat{D}_t = X_t - \inf_{0 \le s\leq t} X_s, \qquad t \ge 0,
\end{equation*}
respectively. The drawdown time is then defined as the first passage time of the drawdown process $D$ reaches a predefined threshold level $a$, i.e., $\tau^X_a = \inf \{ t \ge0: D_t \ge a \}$ and similarly the drawup time is $\widehat{\tau}^X_b = \inf \{ t \ge 0: \widehat{D}_t \ge b \}$. Inspired by \cite{landriault2015frequency-drawdown} and \cite{zhang2015occupation-drawdown}, we focus on the following five drawdown quantities under general Markov models: (1) The drawdown time proceeding the drawup time; (2) The generalized occupation time of the Markov process until the drawdown time; (3) The generalized occupation time of the drawdown process until the drawdown time; (4) The $n$-th drawdown time without recovery; (5) The $n$-th drawdown time with recovery. More specifically, we first approximate the Markov process $X$ by a CTMC and then price the target quantities under the CTMC. The CTMC approximation method has been extensively used for derivative pricing under Markov models. Notable applications include European and barrier options \citep{Mijatovic2013barrier},  American options \citep{eriksson2015american}, Asian options \citep{Song2013weakconvergence,cai2015general,cui2018single}, occupation time derivatives \citep{zhang2022analysis-nonsmooth}, lookback options \citep{zhang2023look}, Parisian options \citep{zhang2023parisian}, autocallable products \citep{cui2024pricing}, and drawdown derivatives \citep{zhang2021american-drawdown,zhang2023drawdown,li2024speed-duration-drawdown}.
Working under the CTMC model, we can establish rigorous characterizations of these five target quantities through dedicated linear system formulations and develop efficient recursive procedures to reduce the complexity of solving these linear systems. For most drawdown quantities, our methodology achieves computational complexity scaling as $\mathcal O(N^3)$ with $N$ CTMC states under general Markov models, while attaining complexity $\mathcal O(N)$ when applied to diffusion models, thereby achieving computational parity with path-independent valuations (e.g., European option pricing) under equivalent model configurations. Finally, we perform rigorous convergence analysis under weak regularity conditions, and conduct extensive numerical experiments to demonstrate the accuracy and efficiency of the proposed algorithms.

Our contributions to the literature are three-fold:
\begin{itemize}
    

    \item We have proposed efficient algorithms for evaluating five important drawdown quantities under general Markov models, overcoming the restrictive L\'evy process/diffusion process/Brownian motion assumptions prevalent in existing literature (e.g., \cite{mijatovic2012drawdown}, \cite{landriault2015frequency-drawdown}, and \cite{zhang2015occupation-drawdown}). In addition, the CTMC approximation has been applied to the analysis of drawdown risk \citep{zhang2023drawdown}, American drawdown option pricing \citep{zhang2021american-drawdown}, and the speed and duration of drawdown \citep{li2024speed-duration-drawdown} under general Markov models, while we extend its application to various highly path-dependent drawdown quantities. Our results are nontrivial extensions of the existing papers. For example, we investigate occupation times of the underlying and drawdown/drawup process up to the drawdown/drawup time, which adds another layer of path-dependency. 

    \item We have established the convergence analysis of the proposed algorithms. The analysis adopts the techniques used in \cite{zhang2023drawdown} but with important and nontrivial modifications to handle the additional path-dependency as mentioned above.

    \item We contribute to ample financial engineering applications of drawdown quantities. Specifically, we price various financial products, including digital options associated with the occupation times of the Markov/drawdown processes and insurance products associated with the frequency of relative drawdowns with/without recovery.
\end{itemize}

The remainder of the paper is structured as follows. In Section \ref{sec:preliminary_results}, we introduce preliminary results of drawdown analysis under the CTMC model. In Section \ref{sec:drawdown_drawup_occupation}, we derive linear systems for five drawdown quantities via the CTMC approximation and propose efficient algorithms for solving these systems. We perform convergence analysis of the proposed algorithms under some mild conditions in Section \ref{sec:convergence_rate}, and conduct numerical experiments to demonstrate the accuracy and efficiency of our algorithms in Section \ref{sec:numerical_experiments}. Finally, we conclude the paper in Section \ref{sec:conclusions}. The construction of the CTMC approximation and proofs are provided in Appendices \ref{sec:ctmc_approximation} and \ref{sec:supplementary_proofs}, respectively.

\section{Preliminaries}
\label{sec:preliminary_results}

Consider a 1D time-homogeneous Markov process $X$ living on an interval $I\subseteq  \mathbb{R}$ with its inﬁnitesimal generator $\mathcal L$ acting on $g\in C^2_c (I)$ as follows:
\begin{align}
    \label{eq:infinitesimal_generatorX}
    \mathcal{L}g(x) = b(x)g'(x) +\frac{1}{2}\sigma^2(x)g''(x) + \int_\mathbb{R} \big(g(x+y) - g(x) - yg'(x) \mathbf{1}_{\{|y| \le 1 \}} \big)\, \nu(x, \diff y),
\end{align}
where $b(x)$ is the drift, $\sigma(x)$ is the diffusion coefficient, and $\nu(x, \diff y)$ is the state-dependent jump measure satisfying $\int_{|y|\leq 1} y^2 \,\nu(x, \diff y)<\infty$ for all $x \in I$.
The class of models characterized by \eqref{eq:infinitesimal_generatorX} is very general as it subsumes diffusions, jump-diffusions, and pure-jump models,
and allows the jump intensity to be state-dependent with ﬁnite or inﬁnite jump activity and ﬁnite or inﬁnite jump variation.

The process $X$ can be approximated by a continuous-time Markov chain (CTMC) $Y$ with the finite state space 
$\mathbb{S} = \{y_0, y_1, \cdots, y_N\},$
where $y_0 < y_1 < \cdots < y_N$ and $y_0$ and $y_N$ are absorbing states.
The details regarding the construction of the CTMC can be found in Appendix \ref{sec:ctmc_approximation}. Denote by $\bm G$ the transition rate matrix of $Y$, whose elements satisfy that $G(x, y) \ge 0$ for $y\in \mathbb{S}\backslash\{x\}$ and $G(x, x) = -\sum_{y\in \mathbb{S}\backslash\{x\}} G(x, y)$ for all $x\in \mathbb{S}$. 
The infinitesimal generator of the CTMC $Y$ is then formulated as:
for any function $g:\mathbb{S} \to \mathbb{R}$,
\begin{align}
    \label{eq:ctmc_operator}
    \mathcal{G}g(x) = \sum_{y \in \mathbb{S}} G(x, y) g(y), \qquad x \in \mathbb{S}.
\end{align}
In the following, for any $x$, define by $x^+=\min \{y\in \mathbb{S},~y>x\}$ the right grid point next to $x$.

\subsection{The First Passage Time of the CTMC}

For any $S \subseteq \mathbb{R}$, the first passage time of the CTMC $Y$ is defined as follows:
\begin{align}
    T_S := \inf \{ t > 0: Y_t \notin S \}.
\end{align}
We consider the following quantity:
\begin{align}
    \label{eq:laplace_transform_first_passage_time}
    P_S(k, x; f) := \mathbb{E}\big[ e^{-\int_{0}^{T_S} k(Y_s) \diff s}  \mathbf{1}_{\{T_S<\infty\}} f(Y_{T_S})  \,\big| \,Y_0 = x \big],
\end{align}
where $k, f:\mathbb{S} \to \mathbb{R}$.

\begin{lemma}\label{passageCTMC}
The quantity $P_S(k, x; f)$ satisfies the following linear system:
\begin{align}\label{eq:linear_system_first_passage_time}
    \begin{cases}
    (k(x) - \mathcal{G})P_S(k, x; f) = 0, \qquad &x \in \mathbb{S} \cap S, \\
    P_S(k, x; f) = f(x), \qquad &x \in \mathbb{S} \backslash S,
    \end{cases}
\end{align}
where $\mathcal{G}$ is the infinitesimal generator of the CTMC specified by \eqref{eq:ctmc_operator}.
\end{lemma}
The proof of Lemma \ref{passageCTMC} is analogous to that of Proposition 3.1 in \cite{zhang2023drawdown} and hence is omitted here. Following \cite{zhang2023drawdown}, the linear system above admits a closed-form solution. More specifically, define ${\pmb P}_S(k; f) = \big( P_S(k, x; f) \big)_{x \in \mathbb{S}},$
then 
\begin{align}
    \label{eq:P_S_closed_form}
    {\pmb P}_S(k; f) = ({\pmb I}_{S^c} + {\pmb I}_S \operatorname{diag}({\pmb k}) - {\pmb I}_S\, {\pmb G})^{-1} {\pmb I}_{S^c} {\pmb f}.
\end{align}
Here, ${\pmb f} = (f(x))_{x \in \mathbb{S}}$, ${\pmb k} = (k(x))_{x \in \mathbb{S}}$, ${\pmb I}_S = (\mathbf{1}_{\{z = x,\, x \in S \}})_{x,\, z \in \mathbb{S}}$, and ${\pmb I}_{S^c} = ( \mathbf{1}_{\{z=x,\, x \in S^c \}})_{x,\,z \in \mathbb{S}}$. It is noteworthy that the dependence of the solution on $f$ is linear.

\begin{remark}
    The Laplace transform of the first passage time $T_S$ can be obtained as a special case by setting $k(x) \equiv q$ and $f(x)\equiv 1$.
\end{remark}

\subsection{The First Passage Time of the Drawdown Process of the CTMC}
\label{subsec:first_passage_time_drawdown}

Write $\overline{Y}_{t} = \sup_{0 \le u \le t} Y_u$ and denote by $\tau_a=\mathrm{inf}\{t\geq 0: \overline{Y}_{t} -Y_t\geq a\}$ the first passage time of the drawdown process of the CTMC $Y$ to the level $a>0$. We call $\tau_a$ the drawdown time. Define
\begin{align}
    Q_a(q, x; f) := \mathbb{E}_x\big[ e^{-q\tau_a} \mathbf{1}_{\{\tau_a<\infty\}} f(Y_{\tau_a}) \big],
\end{align}
where $q \in \mathbb{C}$ with $\Re(q) > 0$, $f:\mathbb{S} \to \mathbb{R}$, and $\mathbb{E}_x[ \cdot ] = \mathbb{E}[ \cdot \,|\, Y_0 = x, \,\overline{Y}_0 = x]$.

\begin{lemma}\label{passage_duration}
    For any $q \in \mathbb{C}$ with $\Re(q) > 0$ and $x \in \mathbb{S}$, the quantity $ Q_a(q, x; f)$ satisfies the following linear system:
\begin{align}
    \label{eq:laplace_transform_tau_a}
  Q_a(q, x; f)=P_{(x-a, x]}(q, x; f(\cdot)\mathbf{1}_{\{\cdot \le x -a\}}) +\sum_{y\in\mathbb{S},\, y>x} P_{(x-a, x]}\big(q, x; \mathbf{1}_{\{\cdot = y \}} \big)Q_a(q, y; f).
\end{align}
\end{lemma}
The proof of Lemma \ref{passage_duration} is very similar to that of Theorem 3.1 in \cite{zhang2023drawdown}
and hence is omitted here. The computationally efficient recursive algorithm proposed by \cite{zhang2023drawdown} can be employed for solving the linear system \eqref{eq:laplace_transform_tau_a}. The key steps of their approach can be summarized as follows: 
\begin{itemize}
    \item Step 1: $Q_a(q, y_N; f) = 0$ as $y_N$ is an absorbing state.
    \item Step 2: Assume $x=y_{\eta_x}$ for some integer $\eta_x$. Calculate $Q_a(q, y_i; f)$ by running the backward from $i={N-1}$ to $i=\eta_x$:
    \begin{itemize}

        \item  Compute the quantities $P_{(y_i-a, y_i]}(q, y_i; f(\cdot)\mathbf{1}_{\{\cdot \le y_i -a\}})$ and $P_{(y_i-a, y_i]}\big(q, y_i; \mathbf{1}_{\{\cdot = y_j \}} \big)$ for all $j>i$ via \eqref{eq:P_S_closed_form};

        \item Compute $Q_a(q, y_i; f)$ via \eqref{eq:laplace_transform_tau_a}. 
    \end{itemize}
\end{itemize}

Significant reduction of the computational cost can be achieved in the following two cases.

Case 1: the Markov process $X$ is a L\'evy process. A CTMC $Y$ with the state space $\mathbb{S} = h \mathbb{Z}$ is chosen to approximate $X$ so that $Y$ itself is also a L\'evy process. Here, $h > 0$ and $\mathbb{Z}$ denotes the set of integers. In this case, $P_S(q, x; f) = P_{S-x}(q, 0; f(x+ \cdot))$ and \eqref{eq:laplace_transform_tau_a} can be rewritten as follows:

\begin{align*}
    Q_a(q, x; f) &=  \sum_{y\in\mathbb{S},\, y\leq x-a} P_{(x-a, x]}\big(q, x; \mathbf{1}_{\{\cdot = y \}} \big)f(y) + \sum_{y\in\mathbb{S}, \,y>x} P_{(x-a, x]}\big(q, x; \mathbf{1}_{\{\cdot = y \}} \big)Q_a(q, y; f)\\
    &=  \sum_{y\in\mathbb{S},\, y\leq -a} P_{(-a, 0]}\big(q, 0; \mathbf{1}_{\{\cdot = y \}} \big)f(y+x) + \sum_{y\in\mathbb{S},\, y>0} P_{(-a, 0]}\big(q, 0; \mathbf{1}_{\{\cdot = y \}} \big)Q_a(q, y+x; f).
\end{align*}
Obviously, it suffices to compute $P_{(-a, 0]}\big(q, 0; \mathbf{1}_{\{\cdot = y_i \}} \big)$ for $y_i\leq -a$ and $y_i>0$ via \eqref{eq:P_S_closed_form} to obtain $Q_a(q,x;f)$ recursively.

Case 2: the Markov process $X$ is a diffusion process. The resulting CTMC is a birth-death process and \eqref{eq:laplace_transform_tau_a} can be simplified as follows:
\begin{align}
    Q_a(q, x; f)
    &= P_{(x-a, x]}(q, x; \mathbf{1}_{\{\cdot = (x -a)^{\ominus}\}}) f((x-a)^{\ominus}) + P_{(x-a, x]}(q, x;  \mathbf{1}_{\{\cdot = x^+ \}})  Q_a(q, x^+; f),
\end{align}
where $(x-a)^\ominus = \sup\{y \in \mathbb{S}: y \le x-a\}$. In this case, $P_{(x-a, x]}(q, x; \mathbf{1}_{\{\cdot = (x -a)^{\ominus}\}})$ and $P_{(x-a, x]}(q, x;  \mathbf{1}_{\{\cdot = x^+ \}})$ can be constructed from two independent solutions $\psi^\pm$ as follows: (see \cite{zhang2023drawdown})  
\begin{align}\label{eq:linear-combination}
    &P_{(x-a, x]}(q, x; \mathbf{1}_{\{\cdot = x^+ \}}) = \frac{\psi^+(q, x) \psi^-(q, (x-a)^{\ominus}) - \psi^+(q, (x-a)^{\ominus}) \psi^-(q, x)}{\psi^+(q, x^+) \psi^-(q, (x-a)^{\ominus}) - \psi^+(q, (x-a)^{\ominus}) \psi^-(q, x^+)}, \\
    &P_{(x-a, x]}(q, x; \mathbf{1}_{\{\cdot = (x -a)^{\ominus}\}}) = \frac{\psi^+(q, x^+) \psi^-(q, x) - \psi^+(q, x) \psi^-(q, x^+)}{\psi^+(q, x^+) \psi^-(q, (x-a)^{\ominus}) - \psi^+(q, (x-a)^{\ominus}) \psi^-(q, x^+)}.
\end{align}
Here, $\psi^\pm$ satisfy the following linear systems:
\begin{align}\label{eq:psi-linear-system}
    (q - \mathcal{G}) \psi^\pm(q, x) = 0, \qquad x \in \mathbb{S} \backslash \{y_0, y_N\},
\end{align}
with boundary conditions: $\psi^+(q, y_0) = \psi^-(q, y_N) = 0$ and $\psi^-(q, y_0) = \psi^+(q, y_N) = 1$.


\section{Drawdowns, Drawups, and Occupation Times}
\label{sec:drawdown_drawup_occupation}

In this section, we derive linear systems for five drawdown quantities
under the CTMC and then propose efficient algorithms for solving these linear systems.

\subsection{The Drawdown Time Preceding the Drawup Time}
\label{subsec:drawdown_preceding_drawup}

Write $\underline{Y}_{t} = \inf_{0 \le u \le t} Y_u$ and denote by $\widehat{\tau}_b=\mathrm{inf}\{t\geq 0: Y_t-\underline{Y}_t\geq b\}$ the first passage time of the drawup process of the CTMC $Y$ to the level $b>0$. $\widehat{\tau}_b$ is called the drawup time. We consider the drawdown time $\tau_a$ and the drawup time $\widehat{\tau}_b$, especially when the former precedes the latter, and define the following quantity:
\begin{align}
    \label{eq:drawdown_preceding_drawup}
    A(q, x, y) := \mathbb{E}_{x, y}\big[ e^{-q\tau_a} \mathbf{1}_{\{\tau_a < \widehat{\tau}_b\}} f(Y_{\tau_a})  \big].
\end{align}
Here, $q \in \mathbb{C}$ with $\Re(q) > 0$, $f:\mathbb{S} \to \mathbb{R}$, and $\mathbb{E}_{x, y}[\cdot] = \mathbb{E}[\cdot \,|\, Y_0 = x, \,\overline{Y}_0 = x, \,\underline{Y}_0 = y]$.

\begin{proposition}
    \label{prop:drawdown_preceding_drawup}
    Suppose $b \ge a$. For any $q \in \mathbb{C}$ with $\Re(q) > 0$,  $x \in \mathbb{S}$, and 
    $y \in \mathbb{S} \cap (x-a, x]$, the quantity $A(q, x, y)$ satisfies the following linear system: 
\begin{equation}\label{eq:linear_system_drawdown_preceding_drawup}
\begin{aligned}
    A(q, x, y) 
    =&~P_{(x-a, x]}(q, x; f(\cdot) \mathbf{1}_{\{\cdot \le x -a\}}) \\
    &~
    + \sum_{z \in\mathbb{S}\cap (x, y+b),\, w \in\mathbb{S}\cap ((x-a)\vee(z-b), y]}R_{(x-a, x]}(q, x, y;  \mathbf{1}_{\{\cdot =z, \,\tilde{\cdot}=w\}})A(q,z,w),
\end{aligned}
\end{equation}
where
\begin{align}
    R_{(\ell, r]}(q, x, y; g) := \mathbb{E} \big[ e^{-qT_r^+} \mathbf{1}_{\{ T_\ell^->T_r^+\}} g(Y_{T_r^+}, \underline{Y}_{T_r^+})\, \big| \, Y_0 = x,\, \underline{Y}_0 = y \big],
\end{align}
which satisfies the following linear system: for $x, y \in \mathbb{S}$ with $x \ge y$,
\begin{align}\label{eq:liner-system-min}
    \begin{cases}
        q R_{(\ell, r]}(q, x, y; g) - \sum\limits_{z \in \mathbb{S}} G(x, z) R_{(\ell, r]}(q, z, y \wedge z; g) = 0, \qquad & x \in (\ell, r], \\
        R_{(\ell, r]}(q, x, y; g) = g(x, y) \mathbf{1}_{\{x > r \}}, \qquad & x \notin (\ell, r].
    \end{cases}
\end{align}
Here, $\ell<r$, $T_r^+ := T_{(-\infty, r]}$, and $T_{\ell}^- := T_{(\ell, \infty)}$.
\end{proposition}
\begin{proof}
The proof of Proposition \ref{prop:drawdown_preceding_drawup} is given in Appendix \ref{proof:drawdown_preceding_drawup}.
\end{proof}

Next, we develop an efficient algorithm to solve the linear systems \eqref{eq:linear_system_drawdown_preceding_drawup} and \eqref{eq:liner-system-min} recursively; see Algorithm \ref{alg:quantity-1}. 

\begin{algorithm}[htbp]
    \caption{Calculate $A(q,x,y)$}\label{alg:quantity-1}
    \begin{algorithmic}
         \State \textbf{Assumptions} $\eta_x:=\{i:y_i=x\}$ and $\eta_y:=\{i:y_i=y\}$. 
         \State \textbf{Notations} $i_a:=\min\{j:y_j>y_i-a\},~i_b:=\min\{j:y_j>y_i-b\}$, $\bm{G}:=\{G(y_m,y_{l})\}_{m,{l}=0}^{N}$, $\bm{I}_{m}$ denotes an ${m\times m}$ identity matrix, and $\text{Diag}(\bm{c})$ denotes the diagonal matrix with diagonal elements being $\bm{c}$.
        \State \textbf{Input} $\mathbb{S},~\bm{G},~f,~\eta_x,~\eta_y,$ and $q$.
        \State \textbf{Initialize} $A(q,y_N,\cdot)\gets 0$
        \For{$i=N-1$ to $\eta_{x}$}
        \State $\bm{G}_i\gets\{G(y_m,y_{l})\}_{m,{l}=i_a}^{i}$,~~~~$\bm{e}_{i}\gets\text{Diag}(\{\mathbf{1}_{\{y_m=y_i\}}\}_{m=i_a}^{i})$,
        \State $\bm{G}_i^-\gets\{\{G(y_m,y_{l})\}_{m=i_a}^{i}\}_{{l}=0}^{i_a-1}$,~~~~$\bm{G}_i^+\gets\{\{G(y_m,y_{l})\}_{m=i_a}^{i}\}_{l=i+1}^{N}$,
        \State $\bm{A}_i^+\gets\{\{A(q,y_m,y_{l})\}_{m=i+1}^{N}\}_{{l}=i_b}^{i_a}$,~~~~$\bm{f}_{i}\gets\{f(y_m)\}_{m=0}^{i_a-1}$,
        \State $\{g_{i,n}(\cdot,\tilde{\cdot})\}_{n=i_a}^{i-1}\gets\bigg\{\sum\limits_{y_i<y_m<y_n+b,\,y_{i_a}\vee y_{m_b}<y_l\leq y_n}\mathbf{1}_{\{\cdot=y_m,\,\tilde{\cdot}=y_l\}}A(q,y_m,y_l)\bigg\}_{n=i_a}^i$,
        \State $P_{(y_i-a,y_i]}(q,y_i;f)\gets -\bm{e}_i(\bm{G}_i-q\bm{I}_{i-i_a+1})^{-1}(\bm{G}_i^-\bm{f}_i)$,
        \State $\{A(q,y_i,y_{l})\}_{{l}=i_b}^{i_a}\gets P_{(y_i-a,y_i]}(q,y_i;f) -\bm{e}_i(\bm{G}_i-q\bm{I}_{i-i_a+1})^{-1}(\bm{G}_i^+\bm{A}_i^+)$,
        \For{$j=i_a+1$ to $i$}
        \State $\bm{G}_{i,j}\gets\{G(y_m,y_{l})\}_{m,{l}=j}^{i}$,~~~~$\bm{e}_{i,j}=\text{Diag}(\{\mathbf{1}_{\{y_m=y_i\}}\}_{m=j}^{i})$,
        \State $\bm{G}_{i,j}^-\gets\{\{G(y_m,y_{l})\}_{m=j}^{i}\}_{{l}=i_a}^{j-1}$,~~~~$\bm{G}_{i,j}^+\gets\{\{G(y_m,y_{l})\}_{m=j}^{i}\}_{{l}=i+1}^{N}$,
        \State $\bm{R}_{i,j}^-\gets\{R_{(y_i-a,y_i]}(q,y_m,y_m;g_{i,m})\}_{m=i_a}^{j-1}$,~~~~$\bm{A}_{i,j}^+\gets\{A(q,y_m,y_j)\}_{m=i+1}^{N}$,
        \State $A(q,y_i,y_j)\gets P_{(y_i-a,y_i]}(q,y_i;f)-\bm{e}_{i,j}(\bm{G}_{i,j}-q\bm{I}_{i-j+1})^{-1}(\bm{G}_{i,j}^-\bm{R}_{i,j}^-+\bm{G}_{i,j}^+\bm{A}_{i,j}^+)$.
        \EndFor
        \EndFor
        \State \Return $A(q,y_{\eta_x},y_{\eta_y})$
    \end{algorithmic}
\end{algorithm}

\begin{remark}\label{rmk:complexity-quantity-1} 
    The complexity of Algorithm \ref{alg:quantity-1} is dominated by the calculations of inverse matrices and matrix multiplications.
    For each fixed $i$, when $j$ ranges from $i_a+1$ to $i$, the matrix $\bm{G}_{i,j}-q\bm{I}_{i-j+1}$ in each recursion step can be expressed as the matrix in the previous step plus a sum of rank-one matrices.
    This observation enables us to utilize the Sherman–Morrison–Woodbury formula \citep{woodbury1950} to efficiently compute $(\bm{G}_{i,j}-q\bm{I}_{i-j+1})^{-1}$ from the inverse matrix in the previous recursion step, 
   see \citet[Section 3.1]{zhang2023drawdown} for more details.  Consequently, the complexity of  calculating inverse matrices in the inner loop is $\mathcal{O}(N^3)$. Furthermore, when $i$ decreases from $N-1$ to $\eta_x$, the Sherman–Morrison–Woodbury formula can be applied again, resulting in a complexity of order $\mathcal{O}(N^4)$ for the whole loops. Consequently, the overall complexity of computing inverse matrices is $\mathcal{O}(N^4)$.

    
    The complexity in the calculations of matrix multiplications is of order $\mathcal{O}(N^3)$ instead of $\mathcal{O}(N^4)$.
    Take matrix multiplications $\bm{G}_{i,j}^- \bm{R}_{i,j}^-$ with $i=N-1,\dots,\eta_x$ and $j=i_a+1,\dots,i$ for example. The computational cost can be reduced as follows. 
    Introduce two new matrices $ \widetilde{\bm{G}}_{i,j+1}^-$ and $ \widetilde{\bm{G}}_{i,j}^-$ by expanding $\bm{G}_{i,j+1}^-$ and $\bm{G}_{i,j}^-$ as follows:
    \begin{equation}
        \widetilde{\bm{G}}_{i,j+1}^-:=
        \begin{bmatrix}
            \bm{G}_{i,j+1}^-\\
            \bm{0}_{1\times (j-i_a+1)}
        \end{bmatrix},\quad\quad\quad
        \widetilde{\bm{G}}_{i,j}^-:=
        \begin{bmatrix}
            \bm{G}_{i,j}^-&\bm{0}_{(i-j+1)\times 1}
        \end{bmatrix},
    \end{equation}
    where $\bm{0}_{m\times n}$ is an $m\times n$ zero matrix.
    By construction, $\widetilde{\bm{G}}_{i,j+1}^--\widetilde{\bm{G}}_{i,j}^-$ is a sparse matrix as only the last row and column contain non-zero elements. The sparsity ensures that
    the calculation of the matrix multiplication $(\widetilde{\bm{G}}_{i,j+1}^--\widetilde{\bm{G}}_{i,j}^-)\bm{R}_{i,j+1}^-$ only requires $\mathcal{O}(N)$ operations. Moreover, notice that
     \begin{equation}
    \begin{aligned}
        \widetilde{\bm{G}}_{i,j+1}^-\bm{R}_{i,j+1}^-=~&(\widetilde{\bm{G}}_{i,j+1}^--\widetilde{\bm{G}}_{i,j}^-)\bm{R}_{i,j+1}^-+\widetilde{\bm{G}}_{i,j}^-\bm{R}_{i,j+1}^-\\
        =~&(\widetilde{\bm{G}}_{i,j+1}^--\widetilde{\bm{G}}_{i,j}^-)\bm{R}_{i,j+1}^-+\bm{G}_{i,j}^-\bm{R}_{i,j}^-.\\
    \end{aligned}
    \end{equation}
    As a result, once $\bm{G}_{i,j}^-\bm{R}_{i,j}^-$ has been calculated, the complexity of evaluating $ \widetilde{\bm{G}}_{i,j+1}^-\bm{R}_{i,j+1}^-$ is $\mathcal{O}(N)$. So is $\bm{G}_{i,j+1}^-\bm{R}_{i,j+1}^-$.  When $j$ increases from $i_a+1$ to $i$ and $i$ decreases from $N-1$ to $\eta_x$, the computations of these matrix multiplications are of $\mathcal{O}(N^3)$ complexity. 

    Therefore, the overall complexity of Algorithm \ref{alg:quantity-1} is $\mathcal{O}(N^4)$. In particular, when the Markov process $X$ is a diffusion process, the overall complexity of Algorithm \ref{alg:quantity-1} can be reduced to $\mathcal{O}(N^2)$. In this case, similar to \eqref{eq:linear-combination}, the quantity
    $A(q,y_i,y_j)$
    can be represented in terms of the vector $\{\psi^\pm(q, y_m)\}_{m=0}^N$, which requires $\mathcal{O}(N)$ operations by solving the linear systems \eqref{eq:psi-linear-system}. Thus, one circumvents the computations of inverse matrices and matrix multiplications, 
    as previously discussed. When $i$ decreases from $N-1$ to $\eta_x$ and $j$ ranges from $i_a+1$ to $i$, the overall complexity is $\mathcal{O}((N-\eta_x)(i-i_a))+O(N)$, which simplifies to $ O(N^2)$.

\end{remark}

\begin{remark}
    For any $q \in \mathbb{C}$ with $\Re(q) > 0$,  $x \in \mathbb{S}$, and 
    $y\in(x-b,x-a]$, the quantity $A(q, x, y)$ satisfies the following simpler linear system: 
    \begin{equation}
        A(q,x,y) = P_{(x-a,x]}(q, x; f(\cdot)\mathbf{1}_{\{\cdot\leq x-a\}}) + \sum_{z\in\mathbb{S}\cap (x, y+b)}P_{(x-a,x]}(q, x; \mathbf{1}_{\{\cdot=z\}})A(q, z, y),
    \end{equation}
    which can be solved in analogous to the quantity $Q_a(q,x;f)$.
\end{remark}

\begin{remark}\label{rmk:drawdown_preceding_drawup_b<a}
    The case $b < a$ is more complicated and we leave it for future work.     
\end{remark}

\subsection{The Generalized Occupation Time of the CTMC until the Drawdown Time}
\label{subsec:occupation_below_level}

Next, we consider the generalized occupation time of the CTMC $Y$ until the drawdown time $\tau_a$ and introduce the following quantity:
\begin{align}
    \label{eq:occupation_below_level}
    B(k, x) = \mathbb{E}_x\big[ e^{-\int_0^{\tau_a} k(Y_s) \diff s}  \mathbf{1}_{\{\tau_a<\infty\}}f(Y_{\tau_a}) \big].
\end{align}
When $k(x) =q\mathbf{1}_{\{x < \xi\}}$ and $f(x)\equiv 1$, $B(k, x)$ is reduced to the Laplace transform of the occupation time of the CTMC $Y$ below the level $\xi$ until the drawdown time $\tau_a$.

\begin{proposition}
    \label{prop:occupation_below_level}
    For $x \in \mathbb{S}$ and any $k: \mathbb{S} \to \mathbb{C}$ such that $\Re(k(x)) \ge 0$ for all $x \in \mathbb{S}$, the quantity $B(k, x)$ satisfies the following linear system: 
\begin{equation}
    \label{eq:linear_system_occupation_below_level}
    \begin{aligned}
    B(k, x) 
    &=P_{(x-a, x]}(k, x; f(\cdot) \mathbf{1}_{\{\cdot \le x -a\}}) +\sum_{y\in\mathbb{S},\,y>x}P_{(x-a, x]}(k, x;  \mathbf{1}_{\{ \cdot =y \}})B(k, y).
    \end{aligned}
\end{equation}
\end{proposition}
\begin{proof}
The proof of Proposition \ref{prop:occupation_below_level} is deferred to Appendix \ref{proof:occupation_below_level}.
\end{proof}

We construct an efficient algorithm to solve the linear system \eqref{eq:linear_system_occupation_below_level}, formally described and analyzed in Algorithm \ref{alg:quantity-2}. 
\begin{algorithm}[htbp]
    \caption{Calculate $B(k,x)$}\label{alg:quantity-2}
    \begin{algorithmic}
        \State \textbf{Assumption} $\eta_x:=\{i:y_i=x\}$.
        \State \textbf{Notations} $i_a:=\min\{j:y_j>y_i-a\}$, $\bm{G}:=\{G(y_m,y_{l})\}_{m,{l}=0}^{N}$, 
        and $\text{Diag}(\bm{c})$ denotes the diagonal matrix with diagonal elements being $\bm{c}$.
        \State \textbf{Input} $\mathbb{S},~\bm{G},~f,~k,$ and $\eta_x$.
        \State \textbf{Initialize} $B(k,y_N)\gets 0$
        \For{$i=N-1$ to $\eta_{x}$}
        \State $\bm{G}_i\gets\{G(y_m,y_{l})\}_{m,{l}=i_a}^{i}$,~~~~$\bm{k}_i\gets\text{Diag}(\{k(y_m)\}_{m=i_a}^i)$,~~~~$\bm{e}_{i}=\text{Diag}(\{\mathbf{1}_{\{y_m=y_i\}}\}_{m=i_a}^{i})$,
        \State $\bm{G}_i^-\gets\{\{G(y_m,y_{l})\}_{m=i_a}^{i}\}_{{l}=0}^{i_a-1}$,~~~~$\bm{G}_i^+\gets\{\{G(y_m,y_{l})\}_{m=i_a}^{i}\}_{{l}=i+1}^{N}$,
        \State $\bm{f}_{i}\gets\{f(y_m)\}_{m=0}^{i_a-1}$,~~~~$\bm{B}_i^+\gets\{B(k,y_m)\}_{m=i+1}^N$,
        \State $B(k,y_i)\gets -\bm{e}_i(\bm{G}_i-\bm{k}_i)^{-1}(\bm{G}_i^-\bm{f}_i+\bm{G}_i^+\bm{B}_i^+)$.
        \EndFor
        \State \Return $B(k,y_{\eta_x})$
    \end{algorithmic}
\end{algorithm}

\begin{remark}\label{rmk:complexity-quantity-2}
      The complexity analysis of Algorithm \ref{alg:quantity-2} is analogous to that of Algorithm \ref{alg:quantity-1}. The only difference is that the former does not involve the running minimum $\underline{Y}$ and one dimension is reduced. Therefore, the complexity of Algorithm \ref{alg:quantity-2} is reduced to $\mathcal{O}(N^3)$ for general Markov processes and $\mathcal{O}(N)$ for diffusion processes.
\end{remark}


\subsection{The Generalized Occupation Time of the Drawdown Process until the Drawdown Time}
\label{subsec:occupation_drawdown_above_level}

In this subsection, we consider the generalized occupation time of the drawdown process until the drawdown time, which is reflected in the following quantity:
\begin{align}
    \label{eq:generalized_occupation_drawdown}
    C(k, x) = \mathbb{E}_x\big[ e^{-\int_0^{\tau_a} k(Y_s, \overline{Y}_s) \diff s} \mathbf{1}_{\{\tau_a<\infty\}} f(Y_{\tau_a}) \big].
\end{align}
When $k(x, y) = q \mathbf{1}_{\{y-x > \xi\}}$ and $f(x)\equiv 1$, $C(k,x)$ is reduced to the Laplace transform of the occupation time of the drawdown process above the level $\xi$ until the drawdown time $\tau_a$.

\begin{proposition}
    \label{prop:generalized_occupation_drawdown}
    For $x\in \mathbb{S}$ and any $k: \mathbb{S}^2 \to \mathbb{C}$ such that $\Re(k(x,y)) \ge 0$ for all $x, y\in\mathbb{S}$, the quantity $C(k, x)$ satisfies the following linear system:
\begin{equation}\label{eq:linear_system_generalized_occupation_drawdown}
    \begin{aligned}
    C(k, x) 
    &=P_{(x-a, x]}(k(\cdot,x), x; f(\cdot) \mathbf{1}_{\{\cdot \le x -a\}}) +\sum_{z\in\mathbb{S},\,z>x} S_{(x-a, x]}(k, x, x;  \mathbf{1}_{\{ \cdot = z \}})C(k, z),
    \end{aligned}
\end{equation}
where
\begin{align}
    S_{(\ell, r]}(k, x, y; g) := \mathbb{E} \big[ e^{-\int_0^{T_r^+} k(Y_s, \overline{Y}_s) \diff s} \mathbf{1}_{\{ T_\ell^->T_r^+ \}} g(Y_{T_r^+})\, \big|\,  Y_0 = x, \overline{Y}_0 = y \big],
\end{align}
which satisfies the following linear system: for $x, y \in \mathbb{S}$ with $x \le y$,
\begin{align}
    \label{eq:linear_system_generalized_occupation_drawdown_r}
    \begin{cases}
        k(x, y) S_{(\ell, r]}(k, x, y; g) - \sum_{z \in \mathbb{S}} G(x, z) S_{(\ell, r]}(k, z, y \vee z; g) = 0, &\qquad x \in (\ell, r], \\
        S_{(\ell, r]}(k, x, y; g) = g(x) \mathbf{1}_{\{x > r \}}, &\qquad x \notin (\ell, r].
    \end{cases}
\end{align}
\end{proposition}
\begin{proof}
The proof of Proposition \ref{prop:generalized_occupation_drawdown} is provided in Appendix \ref{proof:generalized_occupation_drawdown}.
\end{proof}


Next, we develop an efficient algorithm to solve the linear systems \eqref{eq:linear_system_generalized_occupation_drawdown} and \eqref{eq:linear_system_generalized_occupation_drawdown_r} recursively;
see Algorithm \ref{alg:quantity-3}. 

\begin{algorithm}[htbp]
    \caption{Calculate $C(k,x)$}\label{alg:quantity-3}
    \begin{algorithmic}
        \State \textbf{Assumption} $\eta_x:=\{i:y_i=x\}$.
        \State \textbf{Notations} $i_a:=\min\{j:y_j>y_i-a\}$, $\bm{G}:=\{G(y_m,y_{l})\}_{m,{l}=0}^{N}$, 
        and $\text{Diag}(\bm{c})$ denotes the diagonal matrix with diagonal elements being $\bm{c}$.
        \State \textbf{Input} $\mathbb{S},~\bm{G},~f,~k,$ and $\eta_x$.
        \State \textbf{Initialize} $C(k,y_N)\gets 0$
        \For{$i=N-1$ to $\eta_{x}$}
        \State $\bm{G}_i\gets\{G(y_m,y_{l})\}_{m,{l}=i_a}^{i}$,~~~~$\bm{k}_i\gets\text{Diag}(\{k(y_m,y_i)\}_{m=i_a}^i)$,~~~~$\bm{e}_i=\text{Diag}(\{\mathbf{1}_{\{y_m=y_i\}}\}_{m=i_a}^{i})$,
        \State $\bm{G}_i^-\gets\{\{G(y_m,y_{l})\}_{m=i_a}^{i}\}_{{l}=0}^{i_a-1}$,~~~~$\bm{G}_i^+\gets\{\{G(y_m,y_{l})\}_{m=i_a}^{i}\}_{{l}=i+1}^{N}$,
        \State $\bm{f}_{i}\gets\{f(y_m)\}_{m=0}^{i_a-1}$,~~~~$\bm{C}_i^+\gets\{C(k,y_m)\}_{m=i+1}^N$,
        \State $C(k,y_i)\gets -\bm{e}_i(\bm{G}_i-\bm{k}_i)^{-1}(\bm{G}_i^-\bm{f}_i+\bm{G}_i^+\bm{C}_i^+)$.
        \EndFor
        \State \Return $C(k,y_{\eta_x})$
    \end{algorithmic}
\end{algorithm}

\begin{remark}\label{rmk:complexity-quantity-3}
    The complexity of Algorithm \ref{alg:quantity-3} for general Markov processes is $\mathcal{O}(N^3)$, which is the same as that of Algorithm \ref{alg:quantity-2}. However, the complexity can only be reduced to $\mathcal{O}(N^2)$ instead of $\mathcal{O}(N)$ for diffusion processes. This distinction arises because, unlike the formula \eqref{eq:linear-combination}, the coefficient $S_{(x-a,x]}(k,x,x;\mathbf{1}_{\{\cdot=x^+\}})$
    cannot be represented in terms of the functions $\psi^\pm(q, x)$ but needs to be calculated by solving a linear system involving both states $x$ and $y$ that incurs complexity $O(N^2)$. Once the coefficients $P_{(x-a, x]}(k(\cdot,x), x; f(\cdot) \mathbf{1}_{\{\cdot \le x -a\}})$ and $S_{(x-a, x]}(k, x, x;  \mathbf{1}_{\{ \cdot = x^+ \}})$ in the linear system of $C(k, x)$ have been computed, we can solve $C(k, x)$ recursively. The overall complexity of  the resulting Algorithm \ref{alg:quantity-3} for diffusion processes is $\mathcal{O}(N^2)$.
    
\end{remark}


\begin{remark}\label{complexity_Levy}
     When $k(x, y) = q \mathbf{1}_{\{y-x > \xi\}}$, $f(x)\equiv 1$, and the Markov process $X$ is a L\'evy process, we have $C(k,x)\equiv C(k,0)$, which leads to the following simple solution to the linear system \eqref{eq:linear_system_generalized_occupation_drawdown}:
    \begin{equation}\label{eq:quantity-3-Levy}
        C(k,x)=\frac{P_{(-a, 0]}(k(\cdot,0), 0; \mathbf{1}_{\{\cdot \le  -a\}})}{1-S_{(-a, 0]}(k, 0, 0; \mathbf{1}_{\{ \cdot > 0 \}})}.
    \end{equation}
    In this case, one avoids solving the linear systems \eqref{eq:linear_system_generalized_occupation_drawdown} recursively and hence the computational time can be reduced significantly. Denote by $n_a$ the number of grid points within the interval $(-a,0]$. Since the calculations of the quantities $P_{(-a, 0]}(k(\cdot,0), 0; \mathbf{1}_{\{\cdot \le  -a\}})$ and $S_{(-a, 0]}(k, 0, 0; \mathbf{1}_{\{ \cdot > 0 \}})$ both require $\mathcal{O}(n_a^3)$ operations, the complexity in the computation of $C(k,x)$ is $\mathcal{O}(n_a^3)$.
\end{remark}


\begin{remark}
   Our methodology is applicable to evaluate the following quantity associated with the generalized occupation time of the drawup process until the drawdown time:
   \begin{equation} \label{eq:generalized_occupation_drawup}
        \begin{aligned}
            E(k,x,y)=\mathbb{E}_{x,y}\big[ e^{-\int_0^{\tau_a} k(Y_s, \underline{Y}_s) \diff s}\mathbf{1}_{\{\tau_a<\infty\}} f(Y_{\tau_a}) \big],
        \end{aligned}
    \end{equation}
    which is an extension to the quantity considered in \cite{zhang2015occupation-drawdown}.
    Given that the quantity $E(q,x,y)$ also depends on the running minimum $\underline{Y}$, we define two auxiliary quantities:
    \begin{align}
        &U_{(\ell, r]}(k, x, y; g) := \mathbb{E} \big[ e^{-\int_{0}^{T_\ell^-}k(Y_s,\underline{Y}_s)\diff s} \mathbf{1}_{\{ T_\ell^-<T_r^+\}} g(Y_{T_\ell^-}, \underline{Y}_{T_\ell^-})\, \big| \, Y_0 = x,\, \underline{Y}_0 = y \big],\\
        &V_{(\ell, r]}(k, x, y; g) := \mathbb{E} \big[ e^{-\int_{0}^{T_r^+}k(Y_s,\underline{Y}_s)\diff s} \mathbf{1}_{\{ T_\ell^->T_r^+\}} g(Y_{T_r^+}, \underline{Y}_{T_r^+})\, \big| \, Y_0 = x,\, \underline{Y}_0 = y \big].
    \end{align}
   Similar to Proposition \ref{prop:generalized_occupation_drawdown}, we can derive a linear system of $E(k,x,y)$ involving $U_{(\ell, r]}(k, x, y; g)$ and $V_{(\ell, r]}(k, x, y; g)$ and build linear systems for $U_{(\ell, r]}(k, x, y; g)$ and $V_{(\ell, r]}(k, x, y; g)$.

    Also, our methodology can be extended to investigate the occupation time of the drawup process above a level $\xi$ until an independent exponential time \citep{zhang2015occupation-drawdown}, which is reflected in the following quantity:
    \begin{align}
        \label{eq:occupation_drawdown_exponential}
        F(q, x, y) = \mathbb{E}\big[ e^{-q\int_0^{e_p} 1_{\{Y_s-\underline{Y}_s > \xi \}}\diff s} f(Y_{e_p})\,|\, Y_0=x, \underline{Y}_0=y \big].
    \end{align}
    Here, $e_p$ is an exponential random variable with the mean $1/p$ ($p>0$). Similar to Proposition \ref{prop:generalized_occupation_drawdown}, we can build a linear system for $F(q,x,y)$.
\end{remark}

\subsection{The $n$-th Drawdown Time without Recovery}
\label{subsec:nth_drawdown_without_recovery}

Write $\overline{Y}_{s, t} = \sup_{s \le u \le t} Y_u$. The $n$-th drawdown time without recovery $\tilde{\tau}_{a,n}$ is defined as $\tilde{\tau}_{a,n} = \inf \{ t > \tilde{\tau}_{a,n-1}: \overline{Y}_{\tilde{\tau}_{a,n-1}, t} - Y_t \ge a \}$ for $n \ge 1$ with $\tilde{\tau}_{a,0} = 0$. Consider the quantity as follows:
\begin{align}
    \label{eq:nth_drawdown_without_recovery}
    H_n(q, x) = \mathbb{E}_x\big[ e^{-q\tilde{\tau}_{a,n}} f(Y_{\tilde{\tau}_{a,n}}) \big].
\end{align}

\begin{proposition}
    \label{prop:nth_drawdown_without_recovery}
    For $n \ge 1$, any $q \in \mathbb{C}$ with $\Re(q) > 0$, and $x\in \mathbb{S}$, the quantity $H_n(q, x)$ satisfies the following linear system:
\begin{equation}\label{eq:linear_system_nth_drawdown_without_recovery}
    H_n(q, x) =P_{(x-a, x]}(q, x; H_{n-1}(q,\cdot)\mathbf{1}_{\{\cdot \le x -a\}}) +\sum_{y\in\mathbb{S},\, y>x} P_{(x-a, x]}\big(q, x; \mathbf{1}_{\{\cdot = y \}} \big)H_n(q, y)
\end{equation}
with $H_0(q, x) = f(x)$.
\end{proposition}
\begin{proof}
See Appendix \ref{proof:nth_drawdown_without_recovery}.
\end{proof}

Next, we develop an efficient algorithm to solve the linear system \eqref{eq:linear_system_nth_drawdown_without_recovery} recursively from $n = 1$, formally described and analyzed in Algorithm \ref{alg:quantity-4}. 
\begin{algorithm}[htbp]
    \caption{Calculate $H_n(q,x)$}\label{alg:quantity-4}
    \begin{algorithmic}
        \State \textbf{Assumption} $\eta_x:=\{i:y_i=x\}$.
        \State \textbf{Notations} 
        $i_a:=\min\{j:y_j>y_i-a\}$, 
        $\bm{G}:=\{G(y_m,y_{l})\}_{m,{l}=0}^{N}$, $\bm{I}_m$ denotes an $m\times m$ identity matrix, 
        and $\text{Diag}(\bm{c})$ denotes the diagonal matrix with diagonal elements being $\bm{c}$.
        \State \textbf{Input} $\mathbb{S},~\bm{G},~f,~\eta_x,~n,$ and $q$.
        \State \textbf{Initialize} $H_0(q,\cdot)\gets f(\cdot)$
        \For{$k=1$ to $n$}
        \State $H_k(q,y_N)\gets 0$
        \For{$i=N-1$ to $0$}
        \State $\bm{G}_i\gets\{G(y_m,y_{l})\}_{m,{l}=i_a}^{i}$,~~~~$\bm{e}_i=\text{Diag}(\{\mathbf{1}_{\{y_m=y_i\}}\}_{m=i_a}^{i})$,
        \State $\bm{G}_i^-\gets\{\{G(y_m,y_{l})\}_{m=i_a}^{i}\}_{{l}=0}^{i_a-1}$,~~~~ $\bm{G}_i^+\gets\{\{G(y_m,y_{l})\}_{m=i_a}^{i}\}_{{l}=i+1}^{N}$,
        \State $\bm{H}_{k-1,i}^-=\{H_{k-1}(q,y_m)\}_{m=0}^{i_a-1}$,~~~~ $\bm{H}_{k,i}^+=\{H_k(q,y_m)\}_{m=i_a+1}^{N}$,
        \State $H_k(q,y_i)\gets -\bm{e}_i(\bm{G}_i-q\bm{I}_{i-i_a+1})^{-1}(\bm{G}_i^-\bm{H}_{k-1,i}^-+\bm{G}_i^+\bm{H}_{k,i}^+)$.
        \EndFor
        \EndFor
        \State \Return $H_n(q,y_{\eta_x})$
    \end{algorithmic}
\end{algorithm}

\begin{remark}\label{rmk:complexity-quantity-4}
   As revealed from Algorithm \ref{alg:quantity-4}, the complexity analysis for calculating $H_n(q,x)$ is analogous to that of $Q_a(q,x;f)$ except that the former involves an additional recursion with respect to (w.r.t.) $k$. The computations in inverse matrices are still of order $\mathcal{O}(N^3)$ since the inverse matrix $(\bm{G}_i-q\bm{I}_{i-i_a+1})^{-1}$ is independent of the recursion w.r.t. $k$, whereas the matrix multiplications cost $\mathcal{O}(nN^2)$ operations when $k$ increases from $1$ to $n$. Therefore, the overall complexity of Algorithm \ref{alg:quantity-4} for general Markov processes is $\mathcal{O}(\mathrm{max}(n,N)N^2)$. The complexity can be reduced to $\mathcal{O}(nN)$ for diffusion processes.
\end{remark}

    Consider an insurance policy that offers a protection against relative drawdowns. Assume the seller pays the buyer 1 at each relative drawdown time without recovery as long as it occurs before the maturity $T$; see \citet[Section 5]{landriault2015frequency-drawdown}. The price of the insurance policy can be defined as $\sum_{k=1}^{\infty}\mathbb{E}_x\big[e^{-r_f\tilde{\tau}_{a,k}}\mathbf{1}_{\{\tilde{\tau}_{a,k}<T\}}\big]$, whose Laplace transform w.r.t. $T$ is given by
    \begin{align}\label{LaplaceH}
    \int_{0}^{\infty}\sum_{k=1}^{\infty} e^{-qT}\mathbb{E}_x\big[e^{-r_f\tilde{\tau}_{a,k}}\mathbf{1}_{\{\tilde{\tau}_{a,k}<T\}}\big]\,\diff T=\frac{1}{q}\sum_{k=1}^{\infty}\mathbb{E}_x\big[ e^{-(q+r_f)\tilde{\tau}_{a,k}}\big]=\frac{H(q+r_f,x)}{q}.
    \end{align}
    Here, $r_f$ is the risk-free rate and $H(q,x):=\sum_{k=1}^{\infty}\mathbb{E}_x\big[ e^{-q\tilde{\tau}_{a,k}}\big]$.
    The pricing of the insurance product reduces to evaluating $H(q,x)$ as follows.
    
    \begin{corollary}\label{cor:nth_drawdown_without_recovery}
         For any $q \in \mathbb{C}$ with $\Re(q) > 0$ and $x\in \mathbb{S}$, the quantity $H(q, x)$ satisfies the following linear system:
\begin{align}\label{eq:linear_system_drawdown_without_recovery_insurance}
        &H(q,x)\\
       &=~\sum_{y\in\mathbb{S}, \,y\leq x-a}P_{(x-a,x]}(q,x;\mathbf{1}_{\{\cdot=y\}})(1+H(q,y))
        +\sum_{y\in\mathbb{S}, \,y>x}P_{(x-a,x]}(q,x;\mathbf{1}_{\{\cdot=y\}})H(q,y).
    \end{align}
    \end{corollary}

\begin{remark}\label{rmk:complexity-quantity-4-insurance}
    Define $\bm{H}:=\{H(q,x)\}_{x\in\mathbb{S}}$, $\bm{P}:=\{P_{(x-a,x]}(q,x;\mathbf{1}_{\{\cdot=y\}})\mathbf{1}_{\{y\in (-\infty, x-a]\cup (x,\infty)\}}\}_{x,y\in\mathbb{S}}$, and $\widetilde{\bm{P}}:=\{\sum_{y\in\mathbb{S}, \,y\leq x-a}P_{(x-a,x]}(q,x;\mathbf{1}_{\{\cdot=y\}})\}_{x\in\mathbb{S}}$. The linear system \eqref{eq:linear_system_drawdown_without_recovery_insurance} can be written as the following matrix equation: 
    \[(\bm{I}_{N+1}-\bm{P})\bm{H}=\widetilde{\bm{P}},\]
    leading to the following solution:
    \[\bm{H}=(\bm{I}_{N+1}-\bm{P})^{-1}\widetilde{\bm{P}}.\]
   Here, $\bm{I}_{N+1}$ denote an $(N+1)\times(N+1)$ identity matrix.
   Notice that the computation of the matrix $\bm{P}$ requires $\mathcal{O}(N^3)$ operations; see \citet[Section 3.1]{zhang2023drawdown} for more details. Therefore, the overall complexity to solve the linear system \eqref{eq:linear_system_drawdown_without_recovery_insurance} is $\mathcal{O}(N^3)$. 
    If the Markov process $X$ is a L\'evy process, we have $H(q,x)\equiv H(q,0)$,  leading to the following simple solution to the linear system \eqref{eq:linear_system_drawdown_without_recovery_insurance}:
    \begin{equation}\label{eq:linear_system_drawdown_without_recovery_insurance_Levy}
        H(q,x)=\frac{P_{(-a,0]}(q,0;\mathbf{1}_{\{\cdot\leq -a\}})}{1-P_{(-a,0]}(q,0;\mathbf{1}_{\{\cdot\leq -a\}})-P_{(-a,0]}(q,0;\mathbf{1}_{\{\cdot>0\}})}. 
        \end{equation}
        Following the arguments in Remark \eqref{complexity_Levy}, the complexity in the computation of $H(q,x)$ is $\mathcal{O}(n_a^3)$.
\end{remark}

\subsection{The $n$-th Drawdown Time with Recovery}
\label{subsec:nth_drawdown_with_recovery}

The $n$-th drawdown time with recovery $\tau_{a,n}$ is defined as $\tau_{a,n} = \inf \{ t > \tau_{a,n-1}: \overline{Y}_t - Y_t \ge a,\, \overline{Y}_t \geq \overline{Y}_{\tau_{a,n-1}} \}$ for $n \ge 1$ with $\tau_{a,0} = 0$. Consider the following quantity:
\begin{align}
    \label{eq:nth_drawdown_with_recovery}
    J_n(q, x, y) = \mathbb{E}\big[ e^{-q\tau_{a,n}} f(Y_{\tau_{a,n}}, \overline{Y}_{\tau_{a,n}}) \,\big|\, Y_0 = x, \overline{Y}_0 = y \big].
\end{align}

\begin{proposition}
    \label{prop:nth_drawdown_with_recovery}
    For $n \ge 1$, any $q \in \mathbb{C}$ with $\Re(q) > 0$, and $x, y \in \mathbb{S}$ with $x \le y$, $J_n(q, x, y)$ satisfies the following linear system: 
    \begin{equation}\label{eq:linear_system_nth_drawdown_with_recovery}
        \begin{aligned}
            J_n(q, x, y) &= \sum_{z\in\mathbb{S}, \,z\geq y}P_{(-\infty, y)}(q, x; \mathbf{1}_{\{\cdot =z \}})J_n(q, z, z),\\
            J_n(q, y, y) 
            &=  P_{(y-a, y]}(q, y; J_{n-1}(q, \cdot, y) \mathbf{1}_{\{\cdot \le y-a \}})+  \sum_{z\in\mathbb{S},\,z>y}P_{(y-a, y]}(q, y;  \mathbf{1}_{\{\cdot =z \}})J_n(q, z, z).
        \end{aligned}
    \end{equation}
    The initial condition is $J_0(q, x, y) = f(x, y)$.
\end{proposition}
\begin{proof}
The proof of Proposition \ref{prop:nth_drawdown_with_recovery} is given in Appendix \ref{proof:nth_drawdown_with_recovery}.
\end{proof}

The linear system \eqref{eq:linear_system_nth_drawdown_with_recovery} can be solved recursively from $n = 1$ and our algorithm is summarized in Algorithm \ref{alg:quantity-5}. 


\begin{algorithm}[htbp]
    \caption{Calculate $J_n(q,x,y)$}\label{alg:quantity-5}
    \begin{algorithmic}
    \State \textbf{Assumptions} $\eta_x:=\{i:y_i=x\}$ and $\eta_y=\{i: y_i=y\}$.
        \State \textbf{Notations} 
        $i_a:=\min\{j:y_j>y_i-a\}$, $\bm{G}:=\{G(y_m,y_{l})\}_{m,{l}=0}^{N}$, $\bm{I}_m$ denotes an $m\times m$ identity matrix,
        and $\text{Diag}(\bm{c})$ denotes the diagonal matrix with diagonal elements being $\bm{c}$.
        \State \textbf{Input} $\mathbb{S},~\bm{G},~f,~\eta_x,~\eta_y,~n,$ and $q$.
        \State \textbf{Initialize} $J_0(q,\cdot,\tilde{\cdot})\gets f(\cdot,\tilde{\cdot})$
        \For{$k=1$ to $n$}
        \State $J_k(q,\cdot,y_N;f)\gets 0$
        \For{$i=N-1$ to $0$}
        \State $\bm{G}_i\gets\{G(y_m,y_{l})\}_{m,{l}=i_a}^{i}$,~~~~$\bm{e}_i=\text{Diag}(\{\mathbf{1}_{\{y_m=y_i\}}\}_{m=i_a}^{i})$,
        \State $\bm{G}_i^-\gets\{\{G(y_m,y_{l})\}_{m=i_a}^{i}\}_{{l}=0}^{i_a-1}$,~~~~ $\bm{G}_i^+\gets\{\{G(y_m,y_{l})\}_{m=i_a}^{i}\}_{{l}=i+1}^{N}$,
        \State $\widetilde{\bm{G}}_i\gets\{G(y_m,y_{l})\}_{m,{l}=0}^{i-1}$,~~~~ $\widetilde{\bm{G}}_i^+\gets\{\{G(y_m,y_{l})\}_{m=0}^{i-1}\}_{{l}=i}^{N}$,
        \State $\bm{J}_{k-1,i}^-\gets\{J_{k-1}(q,y_m,y_i)\}_{m=0}^{i_a-1}$,~~~~ $\bm{J}_{k,i}^+\gets\{J_k(q,y_m,y_m)\}_{m=i+1}^{N}$,
        \State $J_k(q,y_i,y_i)\gets -\bm{e}_i(\bm{G}_i-q\bm{I}_{i-i_a+1})^{-1}(\bm{G}_i^-\bm{J}_{k-1,i}^-+\bm{G}_i^+\bm{J}_{k,i}^+)$,
        \State $\bm{J}_{k,i-1}^+\gets\{J_k(q,y_m,y_m)\}_{m=i}^{N}$,
        \State$\{J_k(q,y_m,y_i)\}_{m=0}^{i-1}\gets-(\widetilde{\bm{G}}_i-q\bm{I}_{i})^{-1}\widetilde{\bm{G}}_{i}^+\bm{J}_{k,i-1}^+$.
        \EndFor
        \EndFor
        \State \Return $J_n(q,y_{\eta_x},y_{\eta_y})$
    \end{algorithmic}
\end{algorithm}

\begin{remark}\label{rmk:complexity-quantity-5}
    The complexity analysis of Algorithm \ref{alg:quantity-5} is the same as that of Algorithm \ref{alg:quantity-4} and hence is omitted. The complexity of Algorithm \ref{alg:quantity-5} for general Markov processes is $\mathcal{O}(\mathrm{max}(n,N)N^2)$ and can be reduced to $\mathcal{O}(nN)$ for diffusion processes.
\end{remark}

Consider an insurance product that the seller pays the buyer 1 at each relative drawdown time with recovery as long as it occurs before the maturity $T$. The price of the insurance product can be defined as $\sum_{k=1}^{\infty}\mathbb{E}\big[e^{-r_f\tau_{a,k}}\mathbf{1}_{\{\tau_{a,k}<T\}}\,\big|\, Y_0=x, \overline{Y}_0=y\big]$, whose Laplace transform w.r.t. $T$ is given by
\[\int_{0}^{\infty}e^{-qT}\sum_{k=1}^{\infty}\mathbb{E}\big[e^{-r_f\tau_{a,k}}\mathbf{1}_{\{\tau_{a,k}<T\}}\,\big|\, Y_0=x, \overline{Y}_0=y\big]\,\diff T=\frac{1}{q}\sum_{k=1}^{\infty}\mathbb{E}\big[ e^{-(q+r_f)\tau_{a,k}}\,\big|\, Y_0=x, \overline{Y}_0=y\big].\]
We define $J(q,x,y):=\sum_{k=1}^{\infty}\mathbb{E}\big[ e^{-q\tau_{a,k}}\,\big|\, Y_0 = x, \overline{Y}_0 = y\big]$. The pricing of the insurance product reduces to evaluating $J(q,x,y)$ as follows.
\begin{corollary}\label{cor:nth_drawdown_with_recovery}
    For any $q \in \mathbb{C}$ with $\Re(q) > 0$ and $x,y\in \mathbb{S}$, the quantity $J(q, x, y)$ satisfies the following linear system:
    \begin{equation}\label{eq:linear_system_drawdown_with_recovery_insurance}
        \begin{aligned}
            J(q,x,y)=&~\sum_{z\in\mathbb{S},\, z\geq y}P_{(-\infty,y)}(q,x;\mathbf{1}_{\{\cdot =z\}})J(q,z,z),\\
            J(q,y,y)=&~\sum_{z\in\mathbb{S},\, z\leq y-a}P_{(y-a,y]}(q,y;\mathbf{1}_{\{\cdot=z\}})(1+J(q,z,y))\\
            &~+\sum_{z\in\mathbb{S},\, z>y}P_{(y-a,y]}(q,y;\mathbf{1}_{\{\cdot=z\}})J(q,z,z).
        \end{aligned}
    \end{equation}
\end{corollary}

\begin{remark}\label{rmk:complexity-quantity-5-insurance}
    Similar to Remark \ref{rmk:complexity-quantity-4-insurance}, the complexity of solving the linear system \eqref{eq:linear_system_drawdown_with_recovery_insurance} for general Markov processes is $\mathcal{O}(N^3)$.
   If the Markov process $X$ is a L\'evy process, we have $J(q,x,y)\equiv J(q,0,y-x)$,  which leads to the following simplified solution to the linear system \eqref{eq:linear_system_drawdown_with_recovery_insurance}:
    \begin{equation}\label{eq:linear_system_drawdown_with_recovery_insurance_Levy}
        \begin{aligned}
            J(q,x,y)&=P_{(-\infty,y-x)}(q,0;\mathbf{1}_{\{\cdot\geq y-x\}})J(q,0,0),\\
            J(q,0,0)&=\frac{P_{(-a,0]}(q,0;\mathbf{1}_{\{\cdot\leq -a\}})}{1-P_{(-a,0]}(q,0;\mathbf{1}_{\{\cdot>0\}})-P_{(-a,0]}(q,0;P_{(-\infty,0)}(q,\cdot;\mathbf{1}_{\{\tilde{\cdot}\geq 0\}})\mathbf{1}_{\{\cdot\leq -a\}})}.
        \end{aligned}
    \end{equation}
   Since a CTMC $Y$ with the state space $\mathbb{S} = h \mathbb{Z}$ is chosen to approximate the L\'evy process $X$, the quantity $P_{(-\infty,y-x)}(q,0;\mathbf{1}_{\{\cdot\geq y-x\}})$ satisfies an infinitely large linear system. 
   To resolve this difficulty, we select a large $A>a>0$ to approximate $P_{(-\infty,y-x)}(q,0;\mathbf{1}_{\{\cdot\geq y-x\}})$ by $P_{(-A,y-x)}(q,0;\mathbf{1}_{\{\cdot\geq y-x\}})$.  Similarly, let $n_A$ denote the number of grid points within the interval $(-A,0]$, then the complexity in the computation of the quantity $P_{(-A,y-x)}(q,0;\mathbf{1}_{\{\cdot\geq y-x\}})$ is $\mathcal{O}(n_A^3)$,
   which results in $\mathcal{O}(n_A^3)$ order for the overall complexity of solving the linear system \eqref{eq:linear_system_drawdown_with_recovery_insurance_Levy}.
    
If the Markov process $X$ is a diffusion process, the linear system \eqref{eq:linear_system_drawdown_with_recovery_insurance} can be simplified as follows:
    \begin{equation}\label{eq:quantity-5-diffusion-simplified}
        \begin{aligned}
            J(q,x,y)=&~P_{(-\infty,y)}(q,x;\mathbf{1}_{\{\cdot=y\}})J(q,y,y),\\
            J(q,y,y)=&~\frac{P_{(y-a,y]}(q,y;\mathbf{1}_{\{\cdot=y^+\}})J(q,y^+,y^+)+P_{(y-a,y]}(q,y;\mathbf{1}_{\{\cdot= (y-a)^\ominus\}})}{1-P_{(y-a,y]}(q,y;\mathbf{1}_{\{\cdot= (y-a)^\ominus\}})P_{(-\infty,y)}(q,(y-a)^\ominus;\mathbf{1}_{\{\cdot=y\}})}.
        \end{aligned}
    \end{equation}
   In this case, $P_{(-\infty,y)}(q,x;\mathbf{1}_{\{\cdot=y\}})$ can be approximated by $P_{[y_0,y)}(q,x;\mathbf{1}_{\{\cdot=y\}})$, where $y_0$ is a boundary point. Thus, by taking advantage of \eqref{eq:linear-combination}, the linear system \eqref{eq:quantity-5-diffusion-simplified} can be solved recursively with $\mathcal{O}(N)$ computational complexity. 
\end{remark}

\section{Convergence Analysis}
\label{sec:convergence_rate}
In this section, we establish the convergence of the CTMC approximation for five quantities analyzed in Section \ref{sec:drawdown_drawup_occupation}.
Let $\{Y_t^{(n)}\}_{n \in \mathbb{Z}^+}$ be a sequence of CTMCs to approximate $X_t$, where the superscript $(n)$ denotes the number of grid points and $\delta_n$ represents the corresponding mesh size. We assume $\delta_n\to 0$ as $n\to \infty$. Sufficient conditions for $Y^{(n)}\Rightarrow X$ in $D(\mathbb{R})$ are laid out in \cite{Mijatovic2013barrier}, where $\Rightarrow$ stands for the weak convergence of stochastic processes and $D(\mathbb{R})$ is the space of real-valued C$\grave{\text{a}}$dl$\grave{\text{a}}$g functions endowed with the Skorokhod topology. 

Denote by $\omega$ a path of the Markov process $X$ and introduce the following subsets of $D(\mathbb{R})$:
\begin{align*}
    U&:=\{\omega\in D(\mathbb{R}):\inf\{t\geq u:\overline{\omega}_{u,t}-\omega_t\geq a\}=\inf\{t\geq u:\overline{\omega}_{u,t}-\omega_t> a\} ~\text{for all}~u\geq 0\},\\
    V&:=\{\omega\in D(\mathbb{R}):\inf\{t\geq u:\overline{\omega}_{u,t}\geq  L\}=\inf\{t\geq u:\overline{\omega}_{u,t}>L\} ~\text{for all}~u\geq 0\},\\
    \widetilde{\mathcal{W}}&:=\{\omega\in D(\mathbb{R}): \text{for all}~t>s> 0,~\text{if}~\sup_{s\leq u\leq t}(\overline{\omega}_{s,u}-\omega_u)<a,~\text{there exists a}~\delta\in (0, s)~\text{such that} \\
    &\qquad \sup_{s-\delta\leq u\leq t}(\overline{\omega}_{s-\delta,u}-\omega_u)<a\},\\
    \mathcal {W}&:=\{\omega\in D(\mathbb{R}): \text{for all} ~t>u>s\geq 0,~\text{if}~s~\text{and}~u~\text{are consecutive drawdown times with }\\
    & \qquad \text{recovery and}~\overline{\omega}_{s,u}>\overline{\omega}_{u,t},~\text{there exists a}~\delta\in (0,u-s)~\text{such that}~u-\delta\in\mathbb{R}^+\setminus J(\omega)~\text{and}\\
    &\qquad \overline{\omega}_{s,u-\delta}>\overline{\omega}_{u-\delta,t}\}.
\end{align*}
Here, $a>0,~L\in\mathbb{R},~J(\omega):=\{t\geq 0: \omega_t\neq\omega_{t-}\},$ and $\overline{\omega}_{s,t}:=\sup_{s\leq u\leq t}\omega_u$.

We first impose the following assumptions on the paths of $X_t$.
\begin{assumption}[\cite{zhang2023drawdown}]\label{assump:path-x}
    (1) $X_t$ does not have monotone sample paths. (2) $X_t$ is quasi-left continuous, i.e., there exists a sequence of totally inaccessible stopping times that exhausts the jumps of $X_t$. (3) For all $x$, the L\'evy measure $\nu(x,\cdot)$ is either absolutely continuous or if it is discrete, $X_t$ must have a diffusion component, i.e., $\sigma(x)>0$ for all $x$ attainable in the state space. Furthermore, if the L\'evy measure $\nu(x, \cdot)$ has finite activity for any $x$, $X_t$ must have a diffusion component. 
\end{assumption}
Assumption \ref{assump:path-x} ensures that $\tau_a^{(n)}$ converges in distribution to $\tau_a$ as $n\to\infty$; see \citet[Section 4.1]{zhang2023drawdown} for more details. In the following analysis, the quantities bearing the superscript $(n)$ are defined under the CTMC $Y^{(n)}$. 

\begin{theorem}\label{thm:convergence}
    Suppose Assumption \ref{assump:path-x} holds and $Y^{(n)}\Rightarrow X$. 
    \begin{enumerate}[label={(\arabic*)}]
        \item $\Big(\tau_a^{(n)},\widehat{\tau}_b^{(n)},Y_{\tau_a^{(n)}}^{(n)}\Big)$ converges in distribution to $(\tau_a,\widehat{\tau}_b,X_{\tau_a})$ as $n\to \infty$.
        \item Suppose $k(\cdot)$ is a bounded and piecewise continuous univariate function, and there exists a subset $I\subseteq D(\mathbb{R})$ satisfying $P(X\in I)=1$ and for any $\omega\in I$, there exists a finite set $\textbf{T}_{dis}:=\{t_{i}\}_{i=1}^{m}$ such that $k(\omega_s)$ is continuous in $s\in(0,\tau_a)\setminus\textbf{T}_{dis}\subset\mathbb{R}^+\setminus J(\omega)$. Then $\Big(\tau_a^{(n)},\int_{0}^{\tau_a^{(n)}}k(Y_s^{(n)})\diff s,Y_{\tau_a^{(n)}}^{(n)}\Big)$ converges in distribution to $(\tau_a,\int_{0}^{\tau_a}k(X_s)\diff s,X_{\tau_a})$ as $n\to \infty$.
        \item Suppose $k(\cdot,\tilde{\cdot})$ is a bounded and piecewise continuous bivariate function,  and there exists a subset $I\subseteq D(\mathbb{R})$ satisfying $P(X\in I)=1$ and for any $\omega\in I$, there exists a finite set $\textbf{T}_{dis}:=\{t_{i}\}_{i=1}^{m}$ such that $k(\omega_s,\overline{\omega}_s)$ is continuous in $s\in(0,\tau_a)\setminus\textbf{T}_{dis}\subset\mathbb{R}^+\setminus J(\omega)$. Then $\Big(\tau_a^{(n)},\int_{0}^{\tau_a^{(n)}}k(Y_s^{(n)},\overline{Y}_s^{(n)})\diff s,Y_{\tau_a^{(n)}}^{(n)}\Big)$ converges in distribution to $(\tau_a,\int_{0}^{\tau_a}k(X_s,\overline{X}_s)\diff s,X_{\tau_a})$ as $n\to \infty$.
        \item  Suppose $\tilde{\tau}_{a,k}<\tilde{\tau}_{a,k+1}$ almost surely for $k=1,2,\dots$ and $P(X\in U\cap \widetilde{\mathcal W})=1$, then $\Big(\tilde{\tau}_{a,k}^{(n)},Y_{\tilde{\tau}_{a,k}^{(n)}}^{(n)}\Big)$ converges in distribution to $(\tilde{\tau}_{a,k},X_{\tilde{\tau}_{a,k}})$ as $n\to \infty$.
        \item Suppose $\tau_{a,k}<\tau_{a,k+1}$ almost surely for $k=1,2,\dots$ and $P(X\in U\cap V\cap {\mathcal W})=1$, then $\Big(\tau_{a,k}^{(n)},Y_{\tau_{a,k}^{(n)}}^{(n)}\Big)$ converges in distribution to $(\tau_{a,k},X_{\tau_{a,k}})$ as $n\to \infty$.
    \end{enumerate}
\end{theorem}
\begin{proof}
The proof of Theorem \ref{thm:convergence} is deferred to Appendix \ref{proof_convergence}.
\end{proof}


\section{Numerical Experiments}
\label{sec:numerical_experiments}
In this section, we conduct a series of comprehensive numerical experiments to assess the accuracy and efficiency of our proposed five algorithms. Specifically, we price five types of financial derivatives and insurance products
under different Markov models. 
Utilizing the CTMC approximation, the Laplace transforms of these prices can be expressed through the quantities analyzed in Section \ref{sec:drawdown_drawup_occupation}. The prices themselves can be computed by performing Laplace inversions via the Abate-Whitt algorithm \citep{abate1992fourier}.

Let $S_t$ denote the asset price at time $t$ with the interest rate $r_f$ and  dividend yield $d$. We set $S_0=1$, $r_f=0.5$, and $d=0.02$. We consider five quantities associated with the relative drawdown/drawup of the asset price $S_t$, equivalently, the drawdown/drawup of the log-asset price $X_t=\ln S_t$, under the following models:
\begin{itemize}
    \item The Black-Scholes (BS) model \citep{black1973pricing}: 
    \[\diff S_t=(r_f-d)S_t\diff t+\sigma S_t\diff W_t,\]
    where $\{W\}_{t\geq 0}$ is a standard Brownian motion. We set $\sigma=0.3$.
    \item The constant elasticity of variance (CEV) model \citep{cox1996constant}: 
    \[\diff S_t=(r_f-d)S_t\diff t+\sigma S_t^{1+\beta}\diff W_t.\]
    We set $\sigma=0.3$ and $\beta=-0.25$.
    \item The double exponential jump-diffusion model (DEJD) \citep{kou2002jump}: 
    \[\frac{\diff S_t}{S_{t-}}=(r_f-d-\lambda\zeta)\diff t+\sigma\diff W_t+\diff \left(\sum_{i=1}^{N_t}(V_i-1)\right),\]
    where $\{N_t\}_{t\geq 0}$ is a Poisson process with the jump intensity $\lambda$, $\{V_i\}_{i\geq 1}$ is a sequence of i.i.d. random variables with the density of $\ln V_i$ given by $f_{\ln V}(y)=p^+\eta^+e^{-\eta^+ y}\mathbf{1}_{\{y\geq 0\}}+p^-\eta^-e^{\eta^- y}\mathbf{1}_{\{y<0\}}$, and $\zeta=\mathbb{E}[V_i]-1=\frac{p^+\eta^+}{\eta^+-1}+\frac{p^-\eta^-}{\eta^-+1}-1$. We set $\sigma=0.3,~p^+=p^-=0.5,~\eta^+=\eta^-=10,$ and $\lambda=3$.
    \item The variance Gamma (VG) model \citep{madan1998variance}:
    \[S_t=S_0 e^{(r_f-d)t+\theta\gamma_t(1;\nu)+\sigma W_{\gamma_t(1;\nu)}+\ln (1-\theta\nu-\sigma^2\nu/2)/\nu t},\]
    where $\{\gamma_t(1;\nu)\}_{t \geq 0}$ is a Gamma process with the mean 1 and variance $\nu$ that is independent of $\{W_t\}_{t \geq 0}$. We set $\theta=-2.206,~\sigma=0.962,$ and $\nu=0.00254$.
\end{itemize}

Tables \ref{tab:quantity-1}-\ref{tab:quantity-5} present the prices of five types of financial derivatives and insurance products obtained from our five algorithms under the aforementioned four models. To avoid convergence oscillations, we first construct a uniform grid with the step size $h$ over the interval $(x-a,x]$, ensuring that the endpoints $x-a$ and $x$ are grid points. This grid is then extended beyond the interval; further details can be found in \citet[Remark 4.3]{zhang2023drawdown}. 
Denote by $N_x$ the number of grid points within the interval $(x-a,x]$.
The first column lists the values of $N_x$, which serves as a key complexity indicator for our algorithms as revealed from Algorithms \ref{alg:quantity-1}-\ref{alg:quantity-5}. 
The second column provides the benchmarks, obtained by choosing sufficiently wide grid (e.g., $y_0=-4$ and $y_N=4$) and sufficiently large $N_x$ (e.g., $N_x=2^{11}$) for the extrapolation until the first four decimal places of the extrapolated price do not change.
The third to fifth columns display the prices obtained from our algorithms, along with the absolute and relative errors. Moreover, an extrapolation method is adopted to accelerate the convergence of our algorithms, and the resulting prices and the relative errors are shown in the sixth and seventh columns, respectively. Finally, the CPU times are reported in the last column.

Table \ref{tab:quantity-1} shows the numerical results for calculating the probability $P_{x,y}(\tau_a<\widehat{\tau}_b\wedge T):=P(\tau_a<\widehat{\tau}_b\wedge T\,|\, X_0 = x, \,\overline{X}_0 = x, \,\underline{X}_0 = y)$. 
Since this probability relates to the quantity $A(q,x,y)$ with $f\equiv1$ in Section \ref{subsec:drawdown_preceding_drawup} as follows:
$$\int_{0}^{\infty}e^{-qT}P_{x,y}(\tau_a<\widehat{\tau}_b\wedge T)\diff T= \mathbb{E}_{x, y}\big[ e^{-q\tau_a} \mathbf{1}_{\{\tau_a < \widehat{\tau}_b\}} \big]/q=A(q,x,y)/q,$$
it is obtained by performing the Laplace inversion on $A(q,x,y)/q$ w.r.t. $q$ via the Abate-Whitt algorithm. The application of the extrapolation method achieves high accuracy
with the relative error below $1\%$ in less than 30 seconds. The computational times for diffusion models are significantly smaller than those for jump models, which is consistent with the complexity analysis depicted in Remark \ref{rmk:complexity-quantity-1}.  
\renewcommand{\arraystretch}{1.5}
\begin{table}[htbp]
    \caption{Computation results for the probability $P_{x,y}(\tau_a<\widehat{\tau}_b\wedge T)$ under four Markov models. We set $a=0.2$, $b=0.3$, $T=0.5$, $x=\ln 1$, and $y=\ln 0.9$.}\label{tab:quantity-1}
    \vspace{-0.1in}
    \begin{center}
        \begin{tabular}{cccccccc}
             \hline
             $N_x$ & Benchmark & CTMC & Abs. err. & Rel. err. & Extra. & Rel. err.  & Time (sec.)\\ \hline
             \multicolumn{8}{c}{BS model}\\ \hline
             20 & 0.56773 & 0.55212  & 0.01561 & 2.75\% &    &    &\\
             40 & 0.56773 & 0.56000  & 0.00773 & 1.36\% & 0.56789 & 0.03\% &\\
             80 & 0.56773 & 0.56387  & 0.00386 & 0.68\% & 0.56774 & 0.00\% &\\
             160 & 0.56773 & 0.56580  & 0.00193 & 0.34\% & 0.56773 & 0.00\% & 1.63\\ \hline
             \multicolumn{8}{c}{CEV model}\\ \hline
             20 & 0.56690 & 0.55152  & 0.01538 & 2.71\% &    &    &\\
             40 & 0.56690 & 0.55929  & 0.00761 & 1.34\% & 0.56705 & 0.03\% &\\
             80 & 0.56690 & 0.56310  & 0.00380 & 0.67\% & 0.56691 & 0.00\% &\\
             160 & 0.56690 & 0.56500  & 0.00190 & 0.34\% & 0.56690 & 0.00\% & 2.42\\ \hline
             \multicolumn{8}{c}{DEJD model}\\ \hline
             20 & 0.62803 & 0.62049  & 0.00754 & 1.20\% &    &    &\\
             40 & 0.62803 & 0.62428  & 0.00375 & 0.60\% & 0.62806 & 0.00\% &\\
             80 & 0.62803 & 0.62614  & 0.00189 & 0.30\% & 0.62800 & 0.00\% &\\
             160 & 0.62803 & 0.62708  & 0.00095 & 0.15\% & 0.62802 & 0.00\% & 28.95\\ \hline
             \multicolumn{8}{c}{VG model}\\ \hline
             20 & 0.62000 & 0.61624  & 0.00376 & 0.61\% &    &    &\\
             40 & 0.62000 & 0.61867  & 0.00133 & 0.21\% & 0.62110 & 0.18\% &\\
             80 & 0.62000 & 0.61948  & 0.00052 & 0.08\% & 0.62030 & 0.05\% &\\
             160 & 0.62000 & 0.61978  & 0.00022 & 0.04\% & 0.62008 & 0.01\% & 27.36\\ \hline
        \end{tabular}
    \end{center}
\end{table}

Table \ref{tab:quantity-2} displays the results of the digital option with the price
$\mathbb{E}_x\big[e^{-r_f\tau_a}\mathbf{1}_{\{\int_0^{\tau_a}\mathbf{1}_{\{X_s<\xi\}}\diff s <T\}}\big]$. The option holder receives $1$ as long as the occupation time of the log-asset price process below the level $\xi$ until the drawdown time is less than $T$.
The price relates to the quantity $B(k,x)$ with $f\equiv1$ in Section \ref{subsec:occupation_below_level} as follows:
\[\int_0^\infty e^{-qT}\mathbb{E}_x\left[e^{-r_f\tau_a}\mathbf{1}_{\{\int_0^{\tau_a}\mathbf{1}_{\{X_s<\xi\}}\diff s <T\}}\right]\diff T=B(q\mathbf{1}_{\{\cdot<\xi\}}+r_f,x)/q.\]
Similarly, we can obtain the price by performing the Laplace inversion via the Abate-Whitt algorithm. The numerical performances revealed from Table \ref{tab:quantity-2} align with those previously observed from Table \ref{tab:quantity-1}.
\begin{table}[htbp]
    \caption{Computation results for the digital option price $\mathbb{E}_x\big[e^{-r_f\tau_a}\mathbf{1}_{\{\int_0^{\tau_a}\mathbf{1}_{\{X_s<\xi\}}\diff s <T\}}\big]$ under four Markov models. We set $a=0.2,~\xi=0.1,~T=0.5$, and $x=\ln 1$ for the BS, CEV, and DEJD models, and $a=0.5,~\xi=0.2,~T=0.5$ and $x=\ln 1$ for the VG model.}\label{tab:quantity-2}
    \vspace{-0.1in}
    \begin{center}
        \begin{tabular}{cccccccc}
             \hline
             $N_x$ & Benchmark & CTMC & Abs. err. & Rel. err. & Extra. & Rel. err.  & Time (sec.)\\ \hline
             \multicolumn{8}{c}{BS model}\\ \hline
             20 & 0.90338 & 0.89573  & 0.00765 & 0.85\% &    &    &\\
             40 & 0.90338 & 0.89959  & 0.00379 & 0.42\% & 0.90345 & 0.01\% &\\
             80 & 0.90338 & 0.90150  & 0.00188 & 0.21\% & 0.90340 & 0.00\% &\\
             160 & 0.90338 & 0.90244  & 0.00094 & 0.10\% & 0.90339 & 0.00\% & 6.20\\ \hline
             \multicolumn{8}{c}{CEV model}\\ \hline
             20 & 0.90219 & 0.89495  & 0.00724 & 0.80\% &    &    \\
             40 & 0.90219 & 0.89861  & 0.00358 & 0.40\% & 0.90226 & 0.01\% &\\
             80 & 0.90219 & 0.90041  & 0.00178 & 0.20\% & 0.90221 & 0.00\% &\\
             160 & 0.90219 & 0.90130  & 0.00089 & 0.10\% & 0.90220 & 0.00\% & 5.73\\ \hline
             \multicolumn{8}{c}{DEJD model}\\ \hline
             20 & 0.94212 & 0.93743  & 0.00469 & 0.50\% &    &    &\\
             40 & 0.94212 & 0.93978  & 0.00234 & 0.25\% & 0.94214 & 0.00\% &\\
             80 & 0.94212 & 0.94095  & 0.00117 & 0.12\% & 0.94212 & 0.00\% &\\
             160 & 0.94212 & 0.94154  & 0.00058 & 0.06\% & 0.94212 & 0.00\% & 20.13\\ \hline
             \multicolumn{8}{c}{VG model}\\ \hline
             20 & 0.97581 & 0.98441  & 0.00860 & 0.88\% &    &    &\\
             40 & 0.97581 & 0.98056  & 0.00475 & 0.49\% & 0.97671 & 0.09\% &\\
             80 & 0.97581 & 0.97814  & 0.00233 & 0.24\% & 0.97572 & 0.01\% &\\
             160 & 0.97581 & 0.97693  & 0.00112 & 0.11\% & 0.97571 & 0.01\% & 18.53\\ \hline
        \end{tabular}
    \end{center}
\end{table}

In Table \ref{tab:quantity-3}, we consider the digital option with the price $\mathbb{E}_x\big[e^{-r_f\tau_a}\mathbf{1}_{\{\int_0^{\tau_a}\mathbf{1}_{\{\overline{X}_s-X_s>\xi\}}\diff s <T\}}\big]$. The option holder receives $1$ as long as the occupation time of the relative drawdown process above the level $\xi$ until the drawdown time is less than $T$. 
The price of this kind of digital option relates to the quantity $C(k,x)$ with $f\equiv1$ in Section \ref{subsec:occupation_drawdown_above_level} as follows:
\[\int_0^\infty e^{-qT}\mathbb{E}_x\left[e^{-r_f\tau_a}\mathbf{1}_{\{\int_0^{\tau_a}\mathbf{1}_{\{\overline{X}_s-X_s>\xi\}}\diff s <T\}}\right]\diff T=C(q\mathbf{1}_{\{\tilde{\cdot}-\cdot> \xi\}}+r_f,x)/q.\]
The computational time for the L\'evy model can be significantly reduced compared to that of the CEV model, due to the simple solution given in the formula \eqref{eq:quantity-3-Levy}.
\begin{table}[htbp]
    \caption{Computation results for the digital option price $\mathbb{E}_x\big[e^{-r_f\tau_a}\mathbf{1}_{\{\int_0^{\tau_a}\mathbf{1}_{\{\overline{X}_s-X_s>\xi\}}\diff s <T\}}\big]$ under four Markov models. We set $a=0.2,~\xi=0.1,~T=0.1$, and $x=\ln 1$ for the BS, CEV, and DEJD models, and $a=0.5,~\xi=0.2,~T=0.1$, and $x=\ln 1$ for the VG model.}\label{tab:quantity-3}
    \vspace{-0.1in}
    \begin{center}
        \begin{tabular}{cccccccc}
             \hline
             $N_x$ & Benchmark & CTMC & Abs. err. & Rel. err. & Extra. & Rel. err.  & Time (sec.)\\ \hline
             \multicolumn{8}{c}{BS model}\\ \hline
             20 & 0.57770 & 0.62424  & 0.04654 & 8.06\% &  &    &\\
             40 & 0.57770 & 0.60057  & 0.02287 & 3.96\% & 0.57690 & 0.14\% &\\
             80 & 0.57770 & 0.58902  & 0.01132 & 1.96\% & 0.57748 & 0.04\% &\\
             160 & 0.57770 & 0.58333  & 0.00563 & 0.97\% & 0.57764 & 0.01\% & 0.01\\ \hline
             \multicolumn{8}{c}{CEV model}\\ \hline
             20 & 0.57258 & 0.61826  & 0.04568 & 7.98\% &  &    &\\
             40 & 0.57258 & 0.59501  & 0.02243 & 3.92\% & 0.57176 & 0.14\% &\\
             80 & 0.57258 & 0.58368  & 0.01110 & 1.94\% & 0.57236 & 0.04\% &\\
             160 & 0.57258 & 0.57810  & 0.00552 & 0.96\% & 0.57252 & 0.01\% & 8.17\\ \hline
             \multicolumn{8}{c}{DEJD model}\\ \hline
             20 & 0.67087 & 0.71347  & 0.04260 & 6.35\% &  &    &\\
             40 & 0.67087 & 0.69203  & 0.02116 & 3.15\% & 0.67059 & 0.04\% &\\
             80 & 0.67087 & 0.68140  & 0.01053 & 1.57\% & 0.67078 & 0.01\% &\\
             160 & 0.67087 & 0.67612  & 0.00525 & 0.78\% & 0.67085 & 0.00\% & 0.03\\ \hline
             \multicolumn{8}{c}{VG model}\\ \hline
             80 & 0.63236 & 0.65759  & 0.02523 & 3.99\% &  &  &\\
             160 & 0.63236 & 0.64447  & 0.01211 & 1.92\% & 0.63135 & 0.16\% & \\ 
             320 & 0.63236 & 0.63826  & 0.00590 & 0.93\% & 0.63205 & 0.05\% &\\
             640 & 0.63236 & 0.63527  & 0.00291 & 0.46\% & 0.63228 & 0.01\% & 0.36\\ \hline
        \end{tabular}
    \end{center}
\end{table}

Table \ref{tab:quantity-4} presents the numerical results for $\sum_{k=1}^{\infty}\mathbb{E}_x\big[e^{-r_f\tilde{\tau}_{a,k}}\mathbf{1}_{\{\tilde{\tau}_{a,k}<T\}}\big]$, which are the prices of the insurance product associated with the frequency of relative drawdowns without recovery.
The price relates to the quantity $H(q,x)$ in Section \ref{subsec:nth_drawdown_without_recovery} as follows:
\[\int_0^\infty e^{-qT}\sum_{k=1}^{\infty}\mathbb{E}_x\left[e^{-r_f\tilde{\tau}_{a,k}}\mathbf{1}_{\{\tilde{\tau}_{a,k}<T\}}\right]\diff T=H(q+r_f,x)/q.\]
We obtain remarkably accurate prices in several seconds in all cases. 
Notably, for L\'evy models, the prices can be computed by the formula \eqref{eq:linear_system_drawdown_without_recovery_insurance_Levy} within one second.
\begin{table}[htbp]
    \caption{Computation results for the insurance product price $\sum_{k=1}^{\infty}\mathbb{E}_x\big[e^{-r_f\tilde{\tau}_{a,k}}\mathbf{1}_{\{\tilde{\tau}_{a,k}<T\}}\big]$ under four Markov models. We set $a=-\ln(1-0.25),~T=1$, and $x=\ln 1$ for the BS, CEV, and DEJD models, and $a=-\ln(1-0.5),~T=1$, and $x=\ln 1$ for the VG model.}\label{tab:quantity-4}
    \vspace{-0.1in}
    \begin{center}
        \begin{tabular}{cccccccc}
             \hline
             $N_x$ & Benchmark & CTMC & Abs. err. & Rel. err. & Extra. & Rel. err.  & Time (sec.)\\ \hline
             \multicolumn{8}{c}{BS model}\\ \hline
             20 & 0.92475 & 0.87418  & 0.05057 & 5.47\% &    &    &\\
             40 & 0.92475 & 0.89864  & 0.02611 & 2.82\% & 0.92311 & 0.18\% &\\
             80 & 0.92475 & 0.91148  & 0.01327 & 1.43\% & 0.92433 & 0.05\% &\\
             160 & 0.92475 & 0.91807  & 0.00668 & 0.72\% & 0.92465 & 0.01\% & 0.01\\ \hline
             \multicolumn{8}{c}{CEV model}\\ \hline
             20 & 0.94398 & 0.89302  & 0.05096 & 5.40\% &    &    &\\
             40 & 0.94398 & 0.91768  & 0.02630 & 2.79\% & 0.94234 & 0.17\% &\\
             80 & 0.94398 & 0.93062  & 0.01336 & 1.42\% & 0.94356 & 0.04\% &\\
             160 & 0.94398 & 0.93725  & 0.00673 & 0.71\% & 0.94387 & 0.01\% & 1.30\\ \hline
             \multicolumn{8}{c}{DEJD model}\\ \hline
             20 & 1.27996 & 1.22209 & 0.05787 & 4.52\% &    & &\\
             40 & 1.27996 & 1.24972  & 0.03024 & 2.36\% & 1.27735 & 0.20\% &\\
             80 & 1.27996 & 1.26450  & 0.01546 & 1.21\% & 1.27929 & 0.05\% &\\
             160 & 1.27996 & 1.27215  & 0.00781 & 0.61\% & 1.27979 & 0.01\% & 0.03\\ \hline
             \multicolumn{8}{c}{VG model}\\ \hline
             80 & 1.91007 & 1.98464  & 0.07457 & 3.90\% &  &  &\\
             160 & 1.91007 & 1.94389  & 0.03382 & 1.77\% & 1.90314 & 0.36\% & \\ 
             320 & 1.91007 & 1.92584  & 0.01577 & 0.83\% & 1.90778 & 0.12\% & \\ 
             640 & 1.91007 & 1.91763  & 0.00756 & 0.40\% & 1.90943 & 0.03\% & 0.26\\ \hline
        \end{tabular}
    \end{center}
\end{table}

Table \ref{tab:quantity-5} presents the computation results for the prices of the insurance product associated with the frequency of relative drawdowns with recovery.
The price relates to the quantity $J(q,x,y)$ in Section \ref{subsec:nth_drawdown_with_recovery} as follows:
\[\int_0^\infty e^{-qT}\sum_{k=1}^{\infty}\mathbb{E}\left[e^{-r_f\tau_{a,k}}\mathbf{1}_{\{\tau_{a,k}<T\}}\,\Big|\, X_0 = x, \overline{X}_0 = y\right]\diff T=J(q+r_f,x,y)/q.\]
Sufficiently accurate prices are obtained in several seconds. 
Compared to Table \ref{tab:quantity-4}, the efficiency for the CEV model can be improved, due to the simple solution provided by the formula \eqref{eq:quantity-5-diffusion-simplified}.

\begin{table}[htbp]
    \caption{Computation results for the insurance product price $\sum_{k=1}^{\infty}\mathbb{E}\big[e^{-r_f\tau_{a,k}}\mathbf{1}_{\{\tau_{a,k}<T\}}\,\big|\, X_0 = x, \overline{X}_0 = y\big]$ under four Markov models. We set $y=\ln 1$ and the values of other parameters are the same as those in Table \ref{tab:quantity-4}.}
    \label{tab:quantity-5}
    \vspace{-0.1in}
    \begin{center}
        \begin{tabular}{cccccccc}
             \hline
             $N_x$ & Benchmark & CTMC & Abs. err. & Rel. err. & Extra. & Rel. err.  & Time (sec.)\\ \hline
             \multicolumn{8}{c}{BS model}\\ \hline
             20 & 0.68014 & 0.65679  & 0.02335 & 3.43\% &  &    &\\
             40 & 0.68014 & 0.66831  & 0.01183 & 1.74\% & 0.67983 & 0.05\% &\\
             80 & 0.68014 & 0.67419  & 0.00595 & 0.87\% & 0.68006 & 0.01\% &\\
             160 & 0.68014 & 0.67715  & 0.00299 & 0.44\% & 0.68012 & 0.00\% & 0.14\\ \hline
             \multicolumn{8}{c}{CEV model}\\ \hline
             20 & 0.65915 & 0.63702  & 0.02213 & 3.36\% &  &    &\\
             40 & 0.65915 & 0.64793  & 0.01122 & 1.70\% & 0.65883 & 0.05\% &\\
             80 & 0.65915 & 0.65350  & 0.00565 & 0.86\% & 0.65907 & 0.01\% &\\
             160 & 0.65915 & 0.65631  & 0.00284 & 0.43\% & 0.65913 & 0.00\% & 0.33\\ \hline
             \multicolumn{8}{c}{DEJD model}\\ \hline
             20 & 0.82254 & 0.80268  & 0.01986 & 2.41\% &  &    &\\
             40 & 0.82254 & 0.81234  & 0.1020 & 1.24\% & 0.82201 & 0.06\% &\\
             80 & 0.82254 & 0.81738  & 0.00516 & 0.63\% & 0.82241 & 0.02\% &\\
             160 & 0.82254 & 0.81994  & 0.00260 & 0.32\% & 0.82251 & 0.00\% & 0.47\\ \hline
             \multicolumn{8}{c}{VG model}\\ \hline
             80 & 0.98319 & 1.00339  & 0.02020 & 2.05\% &  &  &\\
             160 & 0.98319 & 0.99222  & 0.00903 & 0.92\% & 0.98104 & 0.22\% & \\ 
             320 & 0.98319 & 0.98737  & 0.00418 & 0.43\% & 0.98252 & 0.07\% & \\ 
             640 & 0.98319 & 0.98519  & 0.00200 & 0.20\% & 0.98300 & 0.02\% & 5.43\\ \hline
    \end{tabular}
    \end{center}
\end{table}

Figures \ref{fig:Q1}-\ref{fig:Q5} show the convergence behavior of our proposed algorithms when pricing five types of financial derivatives and insurance products, each corresponding to a quantity analyzed in Section \ref{sec:drawdown_drawup_occupation}, under different Markov models. These figures plot the trend of the logarithm of the absolute error (i.e., $\log_{10}(\epsilon)$) against the logarithm of the number of grid points (i.e., $\log_{10}(N_{x})$). The algorithms exhibit approximately first-order convergence w.r.t. the number of grid points, and the application of extrapolation techniques yields an enhancement in convergence rates.

\begin{figure}[htbp]
    \centering
    \includegraphics[trim=2cm 1cm 2cm 0cm, clip, width=\linewidth]{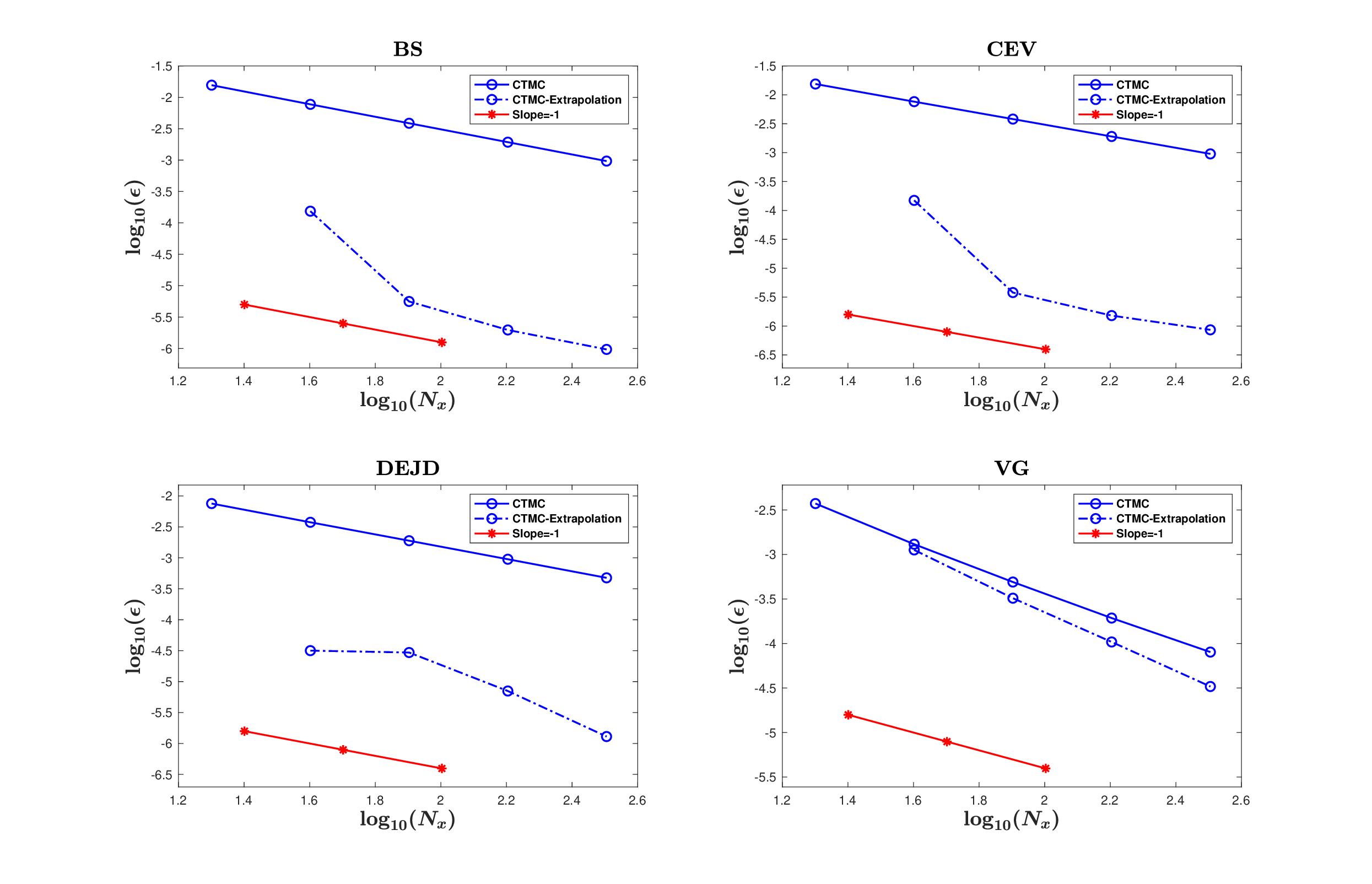}
    \caption{$\log_{10}(\epsilon)$ vs. $\log_{10}(N_x)$ for evaluating $P_{x,y}(\tau_a<\widehat{\tau}_b\wedge T)$ under different models.}
    \label{fig:Q1}
\end{figure}

\begin{figure}[htbp]
    \centering
    \includegraphics[trim=2cm 1cm 2cm 0cm, clip, width=\linewidth]{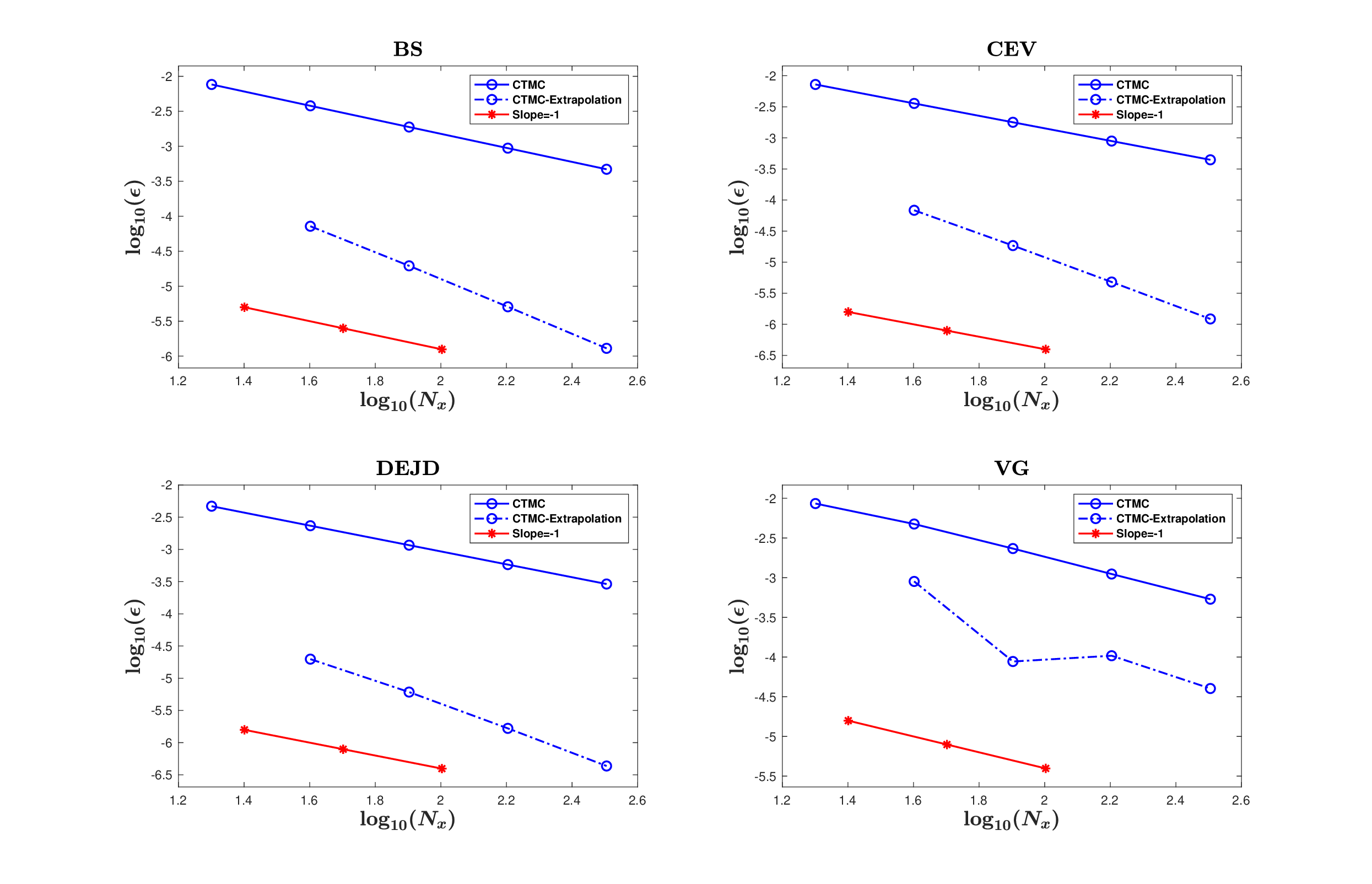}
    \caption{$\log_{10}(\epsilon)$ vs. $\log_{10}({N_x})$ for evaluating $\mathbb{E}_x\Big[e^{-r_f\tau_a}\mathbf{1}_{\{\int_0^{\tau_a}\mathbf{1}_{\{X_s<\xi\}}\diff s <T\}}\Big]$ under different models.}
    \label{fig:Q2}
\end{figure}

\begin{figure}[htbp]
    \centering
    \includegraphics[trim=2cm 1cm 2cm 0cm, clip, width=\linewidth]{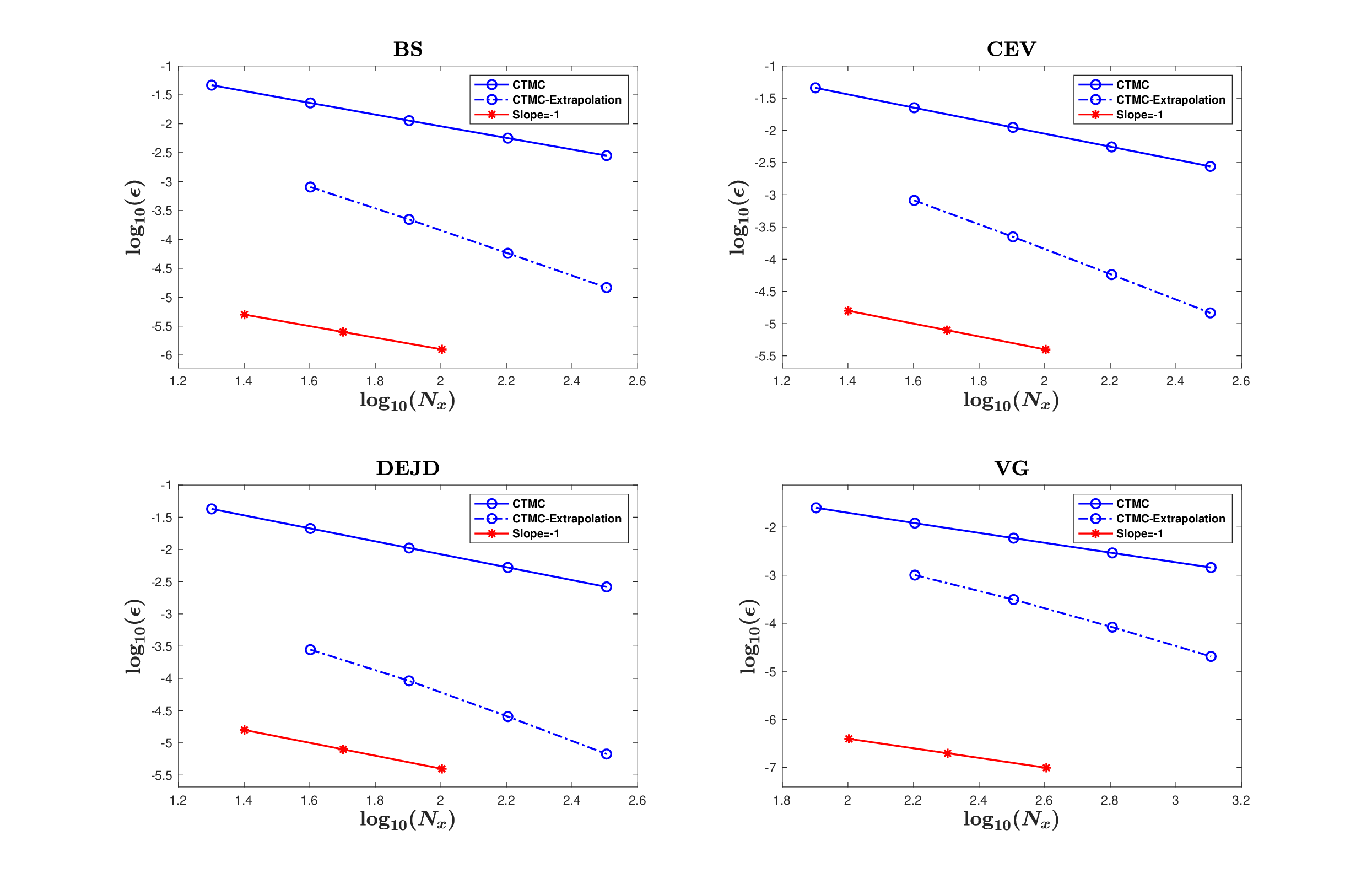}
    \caption{$\log_{10}(\epsilon)$ vs. $\log_{10}(N_x)$ for evaluating $\mathbb{E}_x\Big[e^{-r_f\tau_a}\mathbf{1}_{\{\int_0^{\tau_a}\mathbf{1}_{\{\overline{X}_s-X_s>\xi\}}\diff s <T\}}\Big]$ under different models.}
    \label{fig:Q3}
\end{figure}

\begin{figure}[htbp]
    \centering
    \includegraphics[trim=2cm 1cm 2cm 0cm, clip, width=\linewidth]{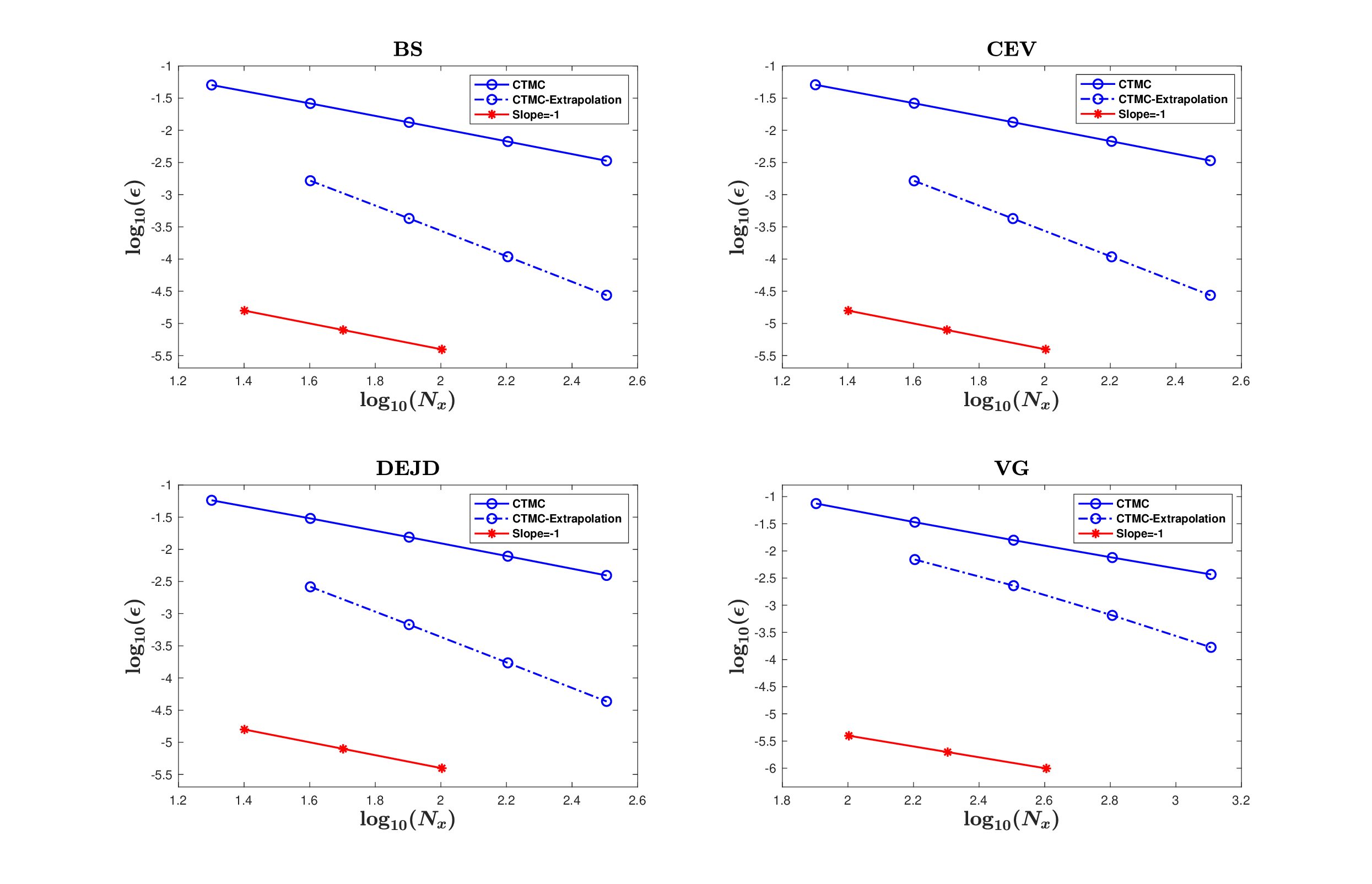}
    \caption{$\log_{10}(\epsilon)$ vs. $\log_{10}(N_x)$ for evaluating $\sum_{k=1}^{\infty}\mathbb{E}_x\Big[e^{-r_f\tilde{\tau}_{a,k}}\mathbf{1}_{\{\tilde{\tau}_{a,k}<T\}}\Big]$ under different models.}
    \label{fig:Q4}
\end{figure}

\begin{figure}[htbp]
    \centering
    \includegraphics[trim=2cm 1cm 2cm 0cm, clip, width=\linewidth]{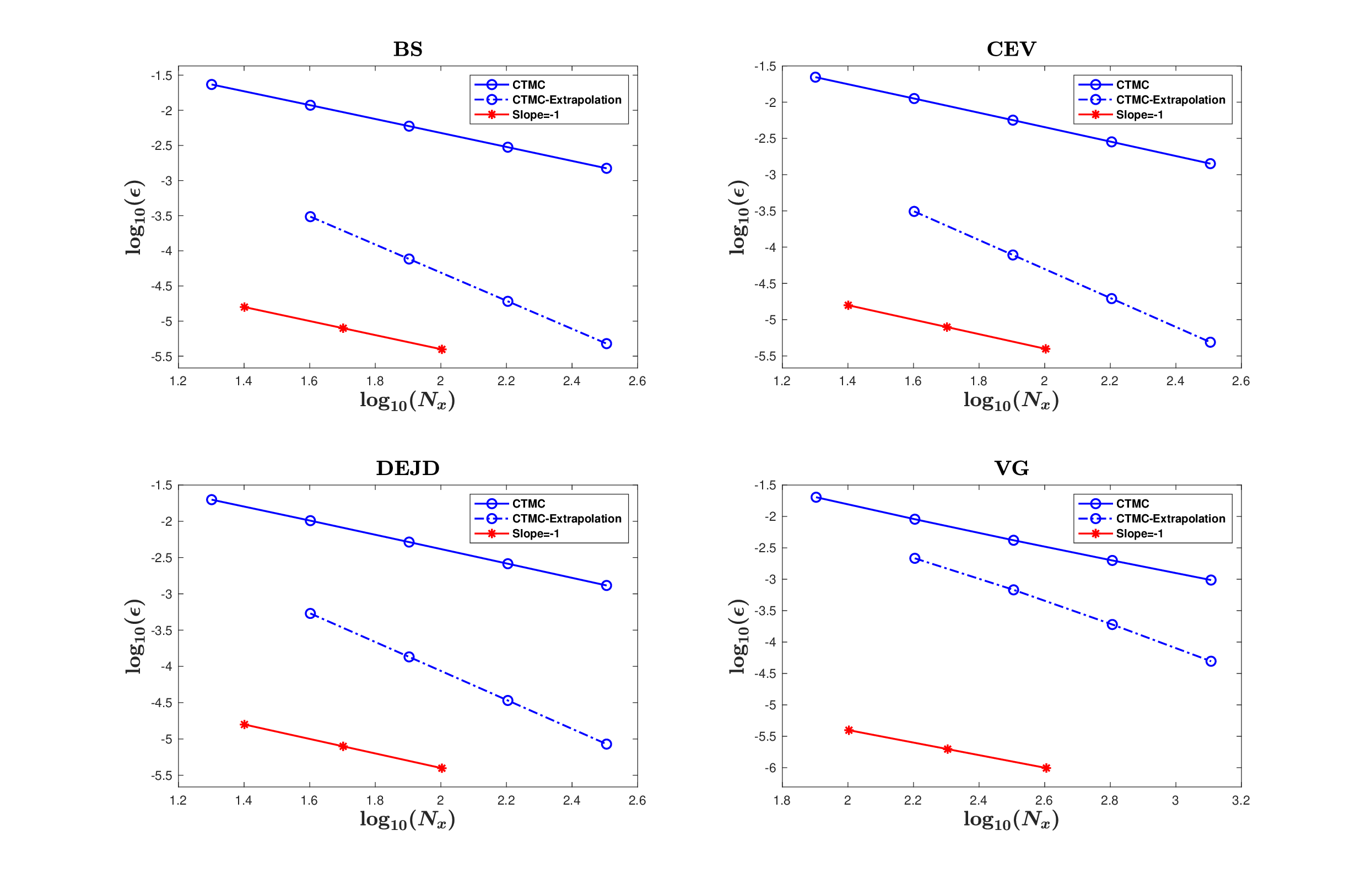}
    \caption{$\log_{10}(\epsilon)$ vs. $\log_{10}(N_x)$ for evaluating  $\mathbb{E}\left[e^{-r_f\tau_{a,k}}\mathbf{1}_{\{\tau_{a,k}<T\}}\,\Big|\, X_0 = x, \overline{X}_0 = y\right]$ under different models.}
    \label{fig:Q5}
\end{figure}

\section{Conclusions}
\label{sec:conclusions}
In this paper, we propose a unified framework for computing five important drawdown quantities under general Markov models. Utilizing the CTMC approximation, we derive linear systems for the associated quantities defined under the CTMC and develop algorithms to solve these systems. In particular, by exploiting the inherent spatial homogeneity of L\'{e}vy models and birth-death characteristic of diffusion models, we construct more efficient algorithms under these two types of models for a majority of quantities. Notably, we circumvent the challenge posed by highly path-dependent quantities and conduct rigorous convergence analysis of our algorithms under mild conditions. Extensive numerical experiments are presented to demonstrate the good performance of our algorithms. Specifically, in most cases, our algorithms achieve remarkable accuracy with relative errors below 1\% within just a few seconds. Even in the worst case, the CPU time remains under 30 seconds. The accuracy can be further enhanced with the aid of the extrapolation technique.

For future research, we can extend our framework in two directions. First, we can consider more complex quantities involving both drawdowns and drawups simultaneously, such as the drawdown time $\tau_a$ preceding the drawup time $\tau_b$ under a more complicated case $b<a$, as discussed in Remark \ref{rmk:drawdown_preceding_drawup_b<a}. 
Second, model assumptions can be generalized beyond Markovian models to incorporate non-Markovian models including rough stochastic local volatility models. The CTMC approximation for rough stochastic local volatility models was constructed in \cite{yang2025general}.

\section*{Acknowledgments}

Pingping Zeng acknowledges the support from the National Natural Science Foundation of China (Grant No. 12171228). Gongqiu Zhang is supported by the National Natural Science Foundation of China (Grant No. 12171408). Weinan Zhang is supported by InnoHK initiative of the ITC of the Hong Kong SAR Government and Laboratory for AI-Powered Financial Technologies.

\section*{Availability of Data and Materials}
The datasets generated during and/or analyzed during the current study are available from the corresponding author on reasonable request.
	 
\section*{Declarations}
We have no conflicts of interest to disclose.

\appendix

\section{Construction of the CTMC Approximation}
\label{sec:ctmc_approximation}
 This section is devoted to constructing a CTMC $Y$ with the state space $\mathbb{S}$ to approximate a Markov process $X$ whose inﬁnitesimal generator $\mathcal L$ is depicted in \eqref{eq:infinitesimal_generatorX}. It suffices to specify the transition rate matrix $\bm{G}$ with elements $G(x, y)$ for $x, y \in \mathbb{S}$, such that the infinitesimal generator $\mathcal{G}$ of $Y$ is an approximation of $\mathcal{L}$ in the sense that:
\begin{align}
    \label{eq:infinitesimal_generator_approximation}
    \mathcal{G}g(x) \approx \mathcal{L}g(x)
\end{align}
for any function $g\in C^2_c (I): I \to \mathbb{R}$ and any $x \in \mathbb{S}$. 

To this end, the derivative terms and the integral term  in \eqref{eq:infinitesimal_generatorX} are approximated by the central difference method and a summation over the discrete state space $\mathbb{S}$, respectively. 
More specifically, the derivative terms are approximated as follows:
\begin{align}
    \label{eq:central_difference_approximation}
    g''(x) &\approx \Delta g(x) := \frac{\nabla^+ g(x) - \nabla^- g(x)}{\delta x},\\
    g'(x) &\approx \nabla g(x) := \frac{\delta^-x}{2\delta x} \nabla^+ g(x) + \frac{\delta^+x}{2\delta x} \nabla^- g(x),
\end{align}
where
\begin{align}
    &x^+ = \inf \{ y \in \mathbb{S}: y > x \}, ~ \delta^+ x = x^+ - x, \qquad &x \in \mathbb{S} \backslash \{y_N\},\\
    &x^- = \sup \{ y \in \mathbb{S}: y < x \}, ~ \delta^- x = x - x^-, \qquad &x \in \mathbb{S} \backslash \{y_0\},\\
    &\delta x=\frac{\delta^+ x +\delta^- x }{2}, \qquad &x \in \mathbb{S} \backslash \{y_0, y_N\},\\
    &\nabla^\pm g(x) = \pm \frac{g(x^\pm) - g(x)}{\delta^\pm x}, \qquad &x \in \mathbb{S} \backslash \{y_0, y_N\}.
\end{align}
On the other hand, the integral term is approximated as follows:
\begin{align}
    \label{eq:jump_approximation}
    &\int_\mathbb{R} \big(g(x+y) - g(x) - yg'(x) \mathbf{1}_{\{|y| \le 1 \}} \big) \nu(x, \diff y) \\
    &\approx~ \frac{1}{2} g''(x)  \int_{ - \delta^-x/2}^{\delta^+x/2} y^2 \nu(x, \diff y) - g'(x) \int_{\mathbb{R} \backslash [-\delta^-x/2, \delta^+x/2]} y\mathbf{1}_{\{|y| \le 1 \}}\nu(x, \diff y)  \\
    &~~~~ + \sum_{z \in \mathbb{S} \backslash \{x\}} \big(g(z) - g(x) \big) \int^{z + \delta^+ z/2 - x + \infty \mathbf{1}_{\{ z= y_N \}}}_{z - \delta^- z/2 - x - \infty \mathbf{1}_{\{ z= y_0 \}}} \nu(x, \diff y),
\end{align}
where the first term approximates the integral for small jumps and the resulting $g''(x)$ is then discretized by the central difference approximation as before. 

For brevity, we introduce
\begin{equation}\label{eq:ctmc-jump-diffusion-notation}
    \begin{aligned}
        &\overline{b}(x)= \int_{\mathbb{R} \backslash [-\delta^-x/2, \delta^+x/2]} y\mathbf{1}_{\{|y| \le 1 \}} \nu(x, \diff y),  \\
        &\overline{\sigma}^2(x) = \frac{1}{2} \int_{ - \delta^-x/2}^{\delta^+x/2} y^2 \nu(x, \diff y), \\
        &\overline{\nu}(x, z)= \int^{z + \delta^+ z/2 - x + \infty 1_{\{ z= y_N \}}}_{z - \delta^- z/2 - x - \infty 1_{\{ z= y_0 \}}} \nu(x, \diff y).
    \end{aligned}
\end{equation}
Thus, the infinitesimal generator $\mathcal{G}$ of $Y$ can be chosen as follows: for $x \in \mathbb{S} \backslash \{y_0, y_N\}$,
\begin{align}
    \mathcal{G} g(x)  = \frac{1}{2} \big( \sigma^2(x) +\overline{\sigma}^2(x) \big) \Delta g(x) + \big(b(x) - \overline{b}(x)\big) \nabla g(x) + \sum_{z \in \mathbb{S} \backslash \{x\}} \overline{\nu}(x, z) \big(g(z) - g(x) \big),
\end{align}
otherwise, $\mathcal{G} g(x) = 0.$
The generator $\mathcal{G}$ implies that the transition rate matrix $G$ satisfies the following conditions:
\begin{align}
& G(x, x^+) = (b(x) - \overline{b}(x)) \frac{\delta^- x}{2\delta^+ x\delta x} + \frac{1}{2} \big( \sigma^2(x) +\overline{\sigma}^2(x) \big) \frac{1}{\delta^+ x\delta x} + \overline{\nu}(x, x^+), ~~~~~~ x \in \mathbb{S} \backslash \{y_0, y_N\},\\
& G(x, x^-) = -(b(x) - \overline{b}(x)) \frac{\delta^+ x}{2\delta^- x \delta x} + \frac{1}{2} \big( \sigma^2(x) +\overline{\sigma}^2(x) \big) \frac{1}{\delta^- x\delta x} + \overline{\nu}(x, x^-), ~~~~ x \in \mathbb{S} \backslash \{y_0, y_N\},\\
& G(x, z) = \overline{\nu}(x, z),~~~~~~~~~~~~~~~~~~~ x \in \mathbb{S}\backslash\{y_0, y_N\},~z \in \mathbb{S} \backslash \{x, x^+, x^-\},\\
& G(x, x) = - \sum_{y \in \mathbb{S} \backslash \{x\}} G(x, y), \qquad x \in \mathbb{S},
\end{align}
and all other elements of $G$ are zeros.

\begin{remark}
    For L\'evy processes, it is desirable to construct a CTMC that preserves the property of independent and stationary increments. To this end, 
 a CTMC $Y$ with the state space $\mathbb{S} = h \mathbb{Z}$ is chosen to approximate the L\'evy process $X$ so that $Y$ itself is also a L\'evy process. Here, $h > 0$ and $\mathbb{Z}$ denotes the set of integers. 
   In this case, for any $x\in\mathbb{S}$, we have
    $x^\pm=x\pm h$ and $\delta^\pm x=\delta x=h.$
    Moreover, $b(x),~\sigma^2(x),$ and $\nu(x,\diff y)$ in \eqref{eq:infinitesimal_generatorX}
    are independent of $x$, 
   which enables us to write $b=b(x),$ $\sigma^2=\sigma^2(x),$ and $\nu(\diff y)=\nu(x,\diff y)$.
Therefore, we introduce
$\overline{b}= \int_{\mathbb{R} \backslash [-h/2, h/2]} y\mathbf{1}_{\{|y| \le 1 \}} \nu(\diff y),$ $\overline{\sigma}^2= \frac{1}{2} \int_{ - h/2}^{h/2} y^2 \nu( \diff y),$  and $\overline{\nu}(z-x)= \int^{z -x + h/2}_{z - x - h/2} \nu( \diff y)$.
It follows that the transition rate matrix $G$ is characterized by
     \begin{align}
         & G(x, x+h) = \frac{b - \overline{b}}{2 h} + \frac{\sigma^2 +\overline{\sigma}^2}{2 h^2} + \overline{\nu}(h), \\
        & G(x, x-h) = -\frac{b - \overline{b}}{2 h} + \frac{\sigma^2 +\overline{\sigma}^2}{2 h^2} + \overline{\nu}(-h),\\
        & G(x, z) = \overline{\nu}(z-x),~~~~ z\in\mathbb{S}\backslash\{x,~x\pm h\},\\
        & G(x, x) = - \sum_{y \in \mathbb{S} \backslash \{x\}} G(x, y).
     \end{align}
\end{remark}

\begin{remark}
    For diffusion processes, we have  $\overline{b}(x) = 0$, $\overline{\sigma}^2(x) = 0$, and $\overline{\nu}(x, z) = 0$ for all $x, z \in \mathbb{S}$. The resulting CTMC is a birth-death process.
\end{remark}

\section{Proofs}
\label{sec:supplementary_proofs}
\subsection{Proof of Proposition \ref{prop:drawdown_preceding_drawup}}\label{proof:drawdown_preceding_drawup}
\begin{proof}
Similar to the proof of Theorem 3.1 in \cite{zhang2023drawdown}, by the introduction of $T_x^+$, $A(q, x, y)$ can be decomposed as follows:
\begin{align}
    A(q, x, y) = \mathbb{E}_{x, y}\big[ e^{-q\tau_a} \mathbf{1}_{\{\tau_a < \widehat{\tau}_b,\, \tau_a < T_x^+ \}} f(Y_{\tau_a})  \big] + \mathbb{E}_{x, y}\big[ e^{-q\tau_a} \mathbf{1}_{\{\tau_a < \widehat{\tau}_b, \,\tau_a > T_x^+ \}} f(Y_{\tau_a})  \big].
\end{align}
On the event $\{\tau_a < T_x^+\}$, we have $\tau_a = T_{x-a}^-$ and the drawup up to $\tau_a$ is less than $b$, leading to $\tau_a < \widehat{\tau}_b$. On the other hand, notice that $\{\tau_a > T_x^+\} = \{T_{x-a}^- > T_x^+\}$. Therefore, for $x \in \mathbb{S}$ and $y \in \mathbb{S} \cap (x-a, x]$, we have
\begin{align}
    &A(q, x, y) \\
    &=\mathbb{E}_{x, y}\big[ e^{-qT_{x-a}^-} \mathbf{1}_{\{T_{x-a}^- < T_x^+ \}} f(Y_{T_{x-a}^-})  \big]\\
    &~~~~+\mathbb{E}_{x, y}\big[ e^{-qT_x^+} \mathbf{1}_{\{T_{x-a}^- > T_x^+ \}}\mathbb{E}[e^{-q(\tau_a-T_x^+)}\mathbf{1}_{\{\tau_a<\widehat{\tau}_b\}}f(Y_{\tau_a})\,|\,Y_{T_x^+}, \overline{Y}_{T_x^+}=Y_{T_x^+}, \underline{Y}_{T_x^+}]\big]\\
    &=  P_{(x-a, x]}(q, x; f(\cdot) \mathbf{1}_{\{\cdot \le x -a\}}) +\mathbb{E}_{x, y}\big[ e^{-qT_x^+} \mathbf{1}_{\{ T_{x-a}^- > T_x^+ \}}A(q, Y_{T_x^+}, \underline{Y}_{T_x^+})  \big]\\
    &=  P_{(x-a, x]}(q, x; f(\cdot) \mathbf{1}_{\{\cdot \le x -a\}}) \\
    &~~~~ +\sum_{z \in \mathbb{S} \cap (x, y + b),\,{w \in\mathbb{S}\cap ((x-a)\vee(z-b), y]}} \mathbb{E}_{x, y}\big[ e^{-qT_x^+} \mathbf{1}_{\{ T_{x-a}^- > T_x^+, \,Y_{T_x^+} = z, \,\underline{Y}_{Y_x^+}=w \}}A(q, Y_{T_x^+}, \underline{Y}_{T_x^+})  \big]\\
    &=P_{(x-a, x]}(q, x; f(\cdot) \mathbf{1}_{\{\cdot \le x -a\}})\\
    &~~~~+\sum_{z \in \mathbb{S} \cap (x, y + b),\,{w \in\mathbb{S}\cap ((x-a)\vee(z-b), y]}} \mathbb{E}_{x, y}\big[ e^{-qT_x^+} \mathbf{1}_{\{ T_{x-a}^- > T_x^+, \,Y_{T_x^+} = z, \,\underline{Y}_{Y_x^+}=w \}}  \big]A(q, z, w),
\end{align}
which in conjunction with the definition of the quantity $R_{(\ell, r]}(q, x, y; g)$ yields (\ref{eq:linear_system_drawdown_preceding_drawup}).

  It remains to derive the linear system for the quantity $R_{(\ell,r]}(q,x,y;g)$. For $x\leq\ell$, we have $\{T_r^+<T_\ell^-\}=\emptyset$, leading to $R_{(\ell, r]}(q,x,y;g)=0$, while for $x>r$, we have $T_r^+=0$, resulting in $R_{(\ell, r]}(q,x,y;g)=g(x,y)$. Consider $x\in(\ell,r]$, inspired by the proof of Proposition 3.1 in \cite{zhang2023drawdown}, we introduce $\mathbb{T}_t=\{kt: k=0,1,\dots\}$ and define
  $\mathbb T_{(\ell,r]}=\inf\{s\in\mathbb{T}_t: Y_s\notin (\ell,r]\},~
  \mathbb T_r^{+}=\mathbb T_{(-\infty,r]},$ and $\mathbb T_\ell^{-}=\mathbb T_{(\ell,\infty)}$. The Markov property of $Y$ yields
\begin{align*}
    &\mathbb{E} \big[ e^{-q\mathbb T_r^{+}} \mathbf{1}_{\{ \mathbb T_\ell^{-}>\mathbb T_r^{+} \}} g(Y_{\mathbb T_r^{+}}, \underline{Y}_{\mathbb T_r^{+}}) \,\big|\,  Y_0 = x, \underline{Y}_0 = y \big]\\
    &=e^{-qt}\sum_{z\in\mathbb{S}}\mathbb{E}\big[ e^{-q\mathbb T_r^{+}} \mathbf{1}_{\{ \mathbb T_\ell^{-}>\mathbb T_r^{+} \}} g(Y_{\mathbb T_r^{+}}, \underline{Y}_{\mathbb T_r^{+}}) \,\big| \, Y_0 = z, \underline{Y}_0 = y \wedge z \big]p(t,x,z)\\
    &=(1-qt)t\sum_{z\in\mathbb{S},\,z\neq x}\mathbb{E}\big[ e^{-q\mathbb T_r^{+}} \mathbf{1}_{\{ \mathbb T_\ell^{-}>\mathbb T_r^{+} \}} g(Y_{\mathbb T_r^{+}}, \underline{Y}_{\mathbb T_r^{+}}) \,\big|\,  Y_0 = z, \underline{Y}_0 = y \wedge z \big]G(x,z)\\
    &~~~~+(1-qt)\mathbb{E}\big[ e^{-q\mathbb T_r^{+}} \mathbf{1}_{\{ \mathbb T_\ell^{-}>\mathbb T_r^{+} \}} g(Y_{\mathbb T_r^{+}}, \underline{Y}_{\mathbb T_r^{+}})\, \big|\,  Y_0 = x, \underline{Y}_0 = y  \big](1+G(x,x)t)+o(t),
\end{align*}
where $p(t,x,z):=P(Y_t=z\,|\,Y_0=x)$. The second equality holds since $p(t,x,z)=\mathbf{1}_{\{x=z\}}+G(x,z)t+o(t)$ and $e^{-qt}=1-qt+o(t)$ when $t\to 0$. 
Dividing both sides by $t$ gives
\begin{align*}
    &\sum_{z\in\mathbb{S}}G(x,z)\mathbb{E}\big[ e^{-q\mathbb T_r^{+}} \mathbf{1}_{\{\mathbb T_r^{+} < \mathbb T_\ell^{-}\}} g(Y_{\mathbb T_r^{+}}, \underline{Y}_{\mathbb T_r^{+}})\, \big|\,  Y_0 = z, \underline{Y}_0 = y \wedge z \big]\\
    &-q\mathbb{E}\big[ e^{-q\mathbb T_r^{+}} \mathbf{1}_{\{\mathbb T_r^{+} < \mathbb T_\ell^{-}\}} g(Y_{\mathbb T_r^{+}}, \underline{Y}_{\mathbb T_r^{+}})\, \big|\,  Y_0 = x, \underline{Y}_0 = y\big]+o(1)=0.
\end{align*}
By similar arguments in the proof of Proposition 3.1 in \cite{zhang2023drawdown}, we finally arrive at \eqref{eq:liner-system-min}.
\end{proof}

\subsection{Proof of Proposition \ref{prop:occupation_below_level}}\label{proof:occupation_below_level}
\begin{proof}
   Recall that on the event $\{\tau_a < T_x^+\}$, we have $\tau_a = T_{x-a}^-$. Moreover, $\{\tau_a > T_x^+\} = \{T_{x-a}^- > T_x^+\}$. Thus,
\begin{align}
   B(k, x)= ~&\mathbb{E}_x\Big[ e^{-\int_0^{\tau_a} k(Y_s) \diff s} \mathbf{1}_{\{\tau_a < T_x^+ \}} f(Y_{\tau_a})  \Big] + \mathbb{E}_x\Big[ e^{-\int_0^{\tau_a} k(Y_s) \diff s} \mathbf{1}_{\{\tau_a > T_x^+ \}} f(Y_{\tau_a})  \Big] \\
   = ~&\mathbb{E}_x\Big[ e^{-\int_0^{T_{x-a}^-} k(Y_s) \diff s} \mathbf{1}_{\{T_{x-a}^- < T_x^+ \}} f(Y_{T_{x-a}^-})  \Big]\\
   ~&+\sum_{z \in \mathbb{S},\, z > x} \mathbb{E}_x\Big[ e^{-\int_0^{T_x^+} k(Y_s) \diff s} \mathbf{1}_{\{T_{x-a}^- > T_x^+, \,Y_{T_x^+}=z\}} B(k, Y_{T_x^+}) \Big]\\
   =~&P_{(x-a, x]}(k, x; f(\cdot) \mathbf{1}_{\{\cdot \le x -a\}})+\sum_{z \in \mathbb{S},\, z > x}\mathbb{E}_x\Big[ e^{-\int_0^{T_x^+} k(Y_s) \diff s} \mathbf{1}_{\{T_{x-a}^- > T_x^+,\, Y_{T_x^+}=z\}} \Big]B(k, z)\\
   =~&P_{(x-a, x]}(k, x; f(\cdot) \mathbf{1}_{\{\cdot \le x -a\}})+\sum_{z \in \mathbb{S},\, z > x}P_{(x-a, x]}(k, x; \mathbf{1}_{\{\cdot =z\}})B(k,z).
\end{align}
\end{proof}

\subsection{Proof of Proposition \ref{prop:generalized_occupation_drawdown}}\label{proof:generalized_occupation_drawdown}
\begin{proof}
    Using similar arguments in the proof of Proposition \ref{prop:occupation_below_level}, we have
    \begin{align}
        C(k, x)=~& \mathbb{E}_x\Big[ e^{-\int_0^{\tau_a} k(Y_s, \overline{Y}_s) \diff s} \mathbf{1}_{\{\tau_a < T_x^+ \}} f(Y_{\tau_a})  \Big] 
         + \mathbb{E}_x\Big[ e^{-\int_0^{\tau_a} k(Y_s, \overline{Y}_s) \diff s} \mathbf{1}_{\{\tau_a > T_x^+ \}} f(Y_{\tau_a})  \Big] \\
        =~& \mathbb{E}_x\Big[ e^{-\int_0^{T_{x-a}^-} k(Y_s, \overline{Y}_s) \diff s} \mathbf{1}_{\{T_{x-a}^- < T_x^+ \}} f(Y_{T_{x-a}^-})  \Big]\\ 
        &+ \sum_{z \in \mathbb{S}, \,z > x} \mathbb{E}_x\Big[ e^{-\int_0^{T_x^+} k(Y_s, \overline{Y}_s) \diff s} \mathbf{1}_{\{T_{x-a}^- > T_x^+, \,Y_{T_x^+}=z \}}C(k, Y_{T_x^+})  \Big]\\
        =~&P_{(x-a, x]}(k(\cdot,x), x; f(\cdot) \mathbf{1}_{\{\cdot \le x -a\}})\\
        &+\sum_{z \in \mathbb{S},\, z > x} \mathbb{E}_x\Big[ e^{-\int_0^{T_x^+} k(Y_s, \overline{Y}_s) \diff s} \mathbf{1}_{\{T_{x-a}^- > T_x^+,\, Y_{T_x^+}=z \}}  \Big]C(k, z)\\
        =~&P_{(x-a, x]}(k(\cdot,x), x; f(\cdot) \mathbf{1}_{\{\cdot \le x -a\}}) + \sum_{z \in \mathbb{S}, \,z > x}S_{(x-a, x]}(k, x, x;  \mathbf{1}_{\{ \cdot =z \}})C(k, z).
    \end{align}
    The linear system for the quantity $S_{(\ell, r]}(k,x,y;g)$ can be derived in a similar way as the quantity $R_{(\ell, r]}(q,x,y;g)$ (see the proof of Proposition \ref{prop:drawdown_preceding_drawup}) and hence is omitted here.
\end{proof}

\subsection{Proof of Proposition \ref{prop:nth_drawdown_without_recovery}}\label{proof:nth_drawdown_without_recovery}
\begin{proof}
   Applying the law of iterated expectations and Lemma \ref{passage_duration} yields
    \begin{align*}
        H_{n}(q,x)=&~\mathbb{E}_x\Big[ e^{-q\tilde{\tau}_{a,n}} f(Y_{\tilde{\tau}_{a,n}}) \Big]\\
        =&~\mathbb{E}_x\Big[ e^{-q\tilde{\tau}_{a,1}}\mathbb{E}\big[e^{-q(\tilde{\tau}_{a,n}-\tilde{\tau}_{a,1})}f(Y_{\tilde{\tau}_{a,n}})\,|\, Y_{\tilde{\tau}_{a,1}},\overline{Y}_{\tilde{\tau}_{a,1}, \tilde{\tau}_{a,1}}\big]  \Big]\\
        =&~\mathbb{E}_x\Big[ e^{-q\tilde{\tau}_{a,1}}H_{n-1}(q,Y_{\tilde{\tau}_{a,1}})  \Big]\\
        =&~Q_a(q,x;H_{n-1}(q,\cdot))\\
        =&~P_{(x-a, x]}(q, x; H_{n-1}(q,\cdot)\mathbf{1}_{\{\cdot \le x -a\}}) +\sum_{y\in\mathbb{S},\, y>x} P_{(x-a, x]}\big(q, x; \mathbf{1}_{\{\cdot = y \}} \big)H_n(q, y).
    \end{align*}
\end{proof}

\subsection{Proof of Proposition \ref{prop:nth_drawdown_with_recovery}}\label{proof:nth_drawdown_with_recovery}
\begin{proof}
    Since $\tau_{a,n}>T_{(-\infty,y)}$ and $\overline{Y}_{T_{(-\infty,y)}}=Y_{T_{(-\infty,y)}}\geq y$, applying the law of iterated expectations yields
    \begin{align*}
        J_n(q,x,y)
        &=\mathbb{E}\Big[ e^{-q T_{(-\infty,y)}}\mathbb{E}\big[e^{-q(\tau_{a,n}-T_{(-\infty,y)})}f(Y_{\tau_{a,n}}, \overline{Y}_{\tau_{a,n}})\,\big|\,Y_{T_{(-\infty,y)}}, \overline{Y}_{T_{(-\infty,y)}}\big]  \,\Big|\, Y_0 = x, \overline{Y}_0 = y \Big]\\
        &=\mathbb{E}\Big[ e^{-q T_{(-\infty,y)}}J_n(q,Y_{T_{(-\infty,y)}},Y_{T_{(-\infty,y)}}) \,\big|\, Y_0 = x, \overline{Y}_0 = y \Big]\\
        &= \sum_{z\in\mathbb{S}, \,z\geq y}P_{(-\infty, y)}(q, x; \mathbf{1}_{\{\cdot =z \}})J_n(q, z, z).
    \end{align*}
    On the other hand, $J_n(q,y,y)$ can be decomposed as follows
    \begin{align*}
        J_n(q,y,y)=&~\mathbb{E}\big[ e^{-q\tau_{a,n}} f(Y_{\tau_{a,n}}, \overline{Y}_{\tau_{a,n}}) \,\big|\, Y_0 = y, \overline{Y}_0 = y \big]\\
        =&~\mathbb{E}\big[ e^{-q\tau_{a,n}} \mathbf{1}_{\{\tau_{a,1}<T_y^+\}}f(Y_{\tau_{a,n}}, \overline{Y}_{\tau_{a,n}}) \,\big|\, Y_0 = y, \overline{Y}_0 = y \big]\\
        &~+\mathbb{E}\big[ e^{-q\tau_{a,n}} \mathbf{1}_{\{\tau_{a,1}>T_y^+\}}f(Y_{\tau_{a,n}}, \overline{Y}_{\tau_{a,n}}) \,\big|\, Y_0 = y, \overline{Y}_0 = y \big]\\
        =&~\mathbb{E}\Big[ e^{-q \tau_{a,1}}\mathbf{1}_{\{\tau_{a,1}<T_y^+\}}\mathbb{E}\big[e^{-q(\tau_{a,n}-\tau_{a,1})}f(Y_{\tau_{a,n}}, \overline{Y}_{\tau_{a,n}})\,\big|\, Y_{\tau_{a,1}}, \overline{Y}_{\tau_{a,1}}\big]  \,\Big|\, Y_0 = y, \overline{Y}_0 = y \Big]\\
        &~+\mathbb{E}\Big[ e^{-q T_y^+}\mathbf{1}_{\{\tau_{a,1}>T_y^+\}}\mathbb{E}\big[e^{-q(\tau_{a,n}-T_y^+)}f(Y_{\tau_{a,n}}, \overline{Y}_{\tau_{a,n}})\,\big|\, Y_{T_y^+}, \overline{Y}_{T_y^+}\big]  \,\Big|\, Y_0 = y, \overline{Y}_0 = y \Big]\\
        =&~\mathbb{E}\Big[ e^{-q T_{y-a}^-}\mathbf{1}_{\{T_{y-a}^-<T_y^+\}}J_{n-1}(q,Y_{T_{y-a}^-},y)  \,\Big|\, Y_0 = y, \overline{Y}_0 = y \Big]\\
        &~+\sum_{z\in\mathbb{S},\, z>y}\mathbb{E}\Big[ e^{-q T_y^+}\mathbf{1}_{\{T_{y-a}^->T_y^+,Y_{T_y^+}=z\}}J_{n}(q,Y_{T_y^+},Y_{T_y^+})  \,\Big|\, Y_0 = y, \overline{Y}_0 = y \Big]\\
        =&~P_{(y-a, y]}(q, y; J_{n-1}(q, \cdot, y) \mathbf{1}_{\{\cdot \le y-a \}})+\sum_{z\in\mathbb{S},\, z>y}P_{(y-a, y]}(q, y; \mathbf{1}_{\{\cdot =z \}})J_n(q,z,z).
    \end{align*}
    Here, the sixth line is attributed to the fact that on the event $\{\tau_{a,1}<T_y^+\}$, we have $\tau_{a,1}=T_{y-a}^-$ and $\overline{Y}_{\tau_{a,1}}=\overline{Y}_0=y$, and the seventh line holds since the condition $Y_0=\overline{Y}_0=y$ implies that $Y_{T_y^+}=\overline{Y}_{T_y^+}>y$.
\end{proof}
\subsection{Proof of Theorem \ref{thm:convergence}}\label{proof_convergence}
\begin{proof}
    For brevity, we write $s'=\tau_a(\omega)$, $s'^{(n)}=\tau_{a}(\omega^{(n)})$, $\overline{\omega}_{s,t}^{(n)}=\sup_{u\in[s,t]}\omega_u^{(n)}$, $MD_{s,t}=\sup_{u\in[s,t]} (\overline{\omega}_{s,u}-\omega_u)$, $MD_{s,t}^{(n)}=\sup_{u\in[s,t]} (\overline{\omega}_{s,u}^{(n)}-\omega_u^{(n)})$, $M_k=\sup_{u\in[\tau_{a,k-1}(\omega),\tau_{a,k}(\omega)]}\omega_{u}$, $M_{k}^{(n)}=\sup_{u\in[\tau_{a,k-1}(\omega^{(n)}), \tau_{a,k}(\omega^{(n)})]}\omega_u^{(n)}$, $T_{M_k}=\inf\{t>\tau_{a,k}(\omega):\omega_t\geq M_k\}$, and $T_{M_k}^{(n)}=\inf\{t>\tau_{a,k}(\omega^{(n)}):\omega_t^{(n)}\geq M_k^{(n)}\},$ where $0\leq s <t$ and $k=1,2,\dots$. Define $\tau_{a,0}(\omega)=\tau_{a,0}(\omega^{(n)})=0$. 
    In the subsequent analysis, we focus on establishing the continuity of the mappings w.r.t. the sample path $\omega$. Combined with the weak convergence $Y^{(n)}\Rightarrow X$, this continuity property ensures 
    the quantities defined under the CTMC $Y^{(n)}$ converge in distribution to their counterparts defined under the Markov process $X$.
    \begin{enumerate}[label={(\arabic*)}]
        \item It suffices to show that $\omega\mapsto(\tau_a(\omega),\widehat{\tau}_b(\omega),\omega_{\tau_a(\omega)})$ is a continuous mapping almost surely. Since the continuity of the mapping $\omega\mapsto(\tau_a(\omega),\omega_{\tau_a(\omega)})$ has been verified by \cite{zhang2023drawdown}, it remains to prove the mapping $\omega\mapsto\widehat{\tau}_b(\omega)$ is continuous. Observe that the running minimum satisfies the identity $\inf_{0\leq s\leq t}X_s=-\sup_{0\leq s\leq t}(-X_s)$. Following the methodology of \citet[Theorem 4.1]{zhang2023drawdown}, we directly verify the continuity of the mapping  $\omega\mapsto\widehat{\tau}_b(\omega)$.

        \item As discussed above, it suffices to show the continuity of the mapping $\omega\mapsto\int_{0}^{\tau_a(\omega)}k(\omega_s)\diff s$. Suppose $\omega^{(n)}\to\omega$ under the Skorokhod topology. For any $\epsilon>0$, denote by $\bm{T}_{\epsilon}=\cup_{1\leq i\leq m}(t_{i}-\epsilon,t_{i}+\epsilon)$ the neighbourhood of $\bm{T}_{dis}$. By \citet[Chapter VI, Section 2, 2.3]{jacod2013limit}, for $s\in(0,\tau_a)\setminus\bm{T}_{\epsilon}$, the convergence $\omega_s^{(n)}\to\omega_s$ holds as $n\to\infty$, ensuring the convergence $k(\omega_s^{(n)})\to k(\omega_s)$. Denote by $U_k$ the upper bound of the function $k(\cdot)$. Notice that $\tau_{a}(\omega^{(n)})\to\tau_a(\omega)$ as $n\to\infty$, then it follows that \[\left|\int_{0}^{\tau_a(\omega)}k(\omega_s^{(n)})\diff s-\int_{0}^{\tau_a(\omega^{(n)})}k(\omega_s^{(n)})\diff s\right|\leq U_k\left|\tau_{a}(\omega)- \tau_{a}(\omega^{(n)})\right|\to 0.\]
        On the other hands,
        \begin{align*}
            &\left|\int_{0}^{\tau_a(\omega)}k(\omega_s)\diff s-\int_{0}^{\tau_a(\omega)}k(\omega_s^{(n)})\diff s\right|\\
            &\leq\int_{\bm{T}_{\epsilon}}\left|k(\omega_s)-k(\omega_s^{(n)})\right|\diff s+\int_{(0,\tau_a(\omega))\setminus\bm{T}_{\epsilon}}\left|k(\omega_s)-k(\omega_s^{(n)})\right|\diff s\\
            &\leq 4mU_k\epsilon+\int_{(0,\tau_a(\omega))\setminus\bm{T}_{\epsilon}}\left|k(\omega_s)-k(\omega_s^{(n)})\right|\diff s\\
            &\to 4mU_k\epsilon.
        \end{align*}
        The second inequality is established by the boundedness of the function $k(\cdot)$ and the last line follows from an application of the dominated convergence theorem.
        Since $\epsilon$ can be arbitrarily small, combining the above two results  yields that $\int_{0}^{\tau_a(\omega^{(n)})}k(\omega_s^{(n)})\diff s\to\int_{0}^{\tau_a(\omega)}k(\omega_s)\diff s$ as $n\to\infty$.
        \item The proof of the continuity of the mapping $\omega\mapsto\int_{0}^{\tau_a(\omega)}k(\omega_s,\overline{\omega}_s)\diff s$ is analogous to that of the mapping $\omega\mapsto\int_{0}^{\tau_a(\omega)}k(\omega_s)\diff s$ and hence is omitted.
        \item According to \citet[Theorem 4.1]{zhang2023drawdown}, $\tilde{\tau}_{a,1}(\omega^{(n)})=\tau_a(\omega^{(n)})\to\tau_a(\omega)=\tilde{\tau}_{a,1}(\omega)$ as $n\to\infty$. 
        Therefore, it suffices to show the continuity of the mapping $\omega\mapsto\tilde{\tau}_{a,2}(\omega)$, and the continuity of the mappings $\omega\mapsto\tilde{\tau}_{a,k}$ for $k=3,\dots$ can be proved by mathematical induction. 

        ~~~~Assume $s'<\tilde{\tau}_{a,2}(\omega)<s,~\omega\in U$, and $s\in\mathbb{R}^+\setminus J(\omega)$, then $MD_{s',s}>a$. 
        By the right-continuity of the path $\omega$ at $s'$, $MD_{s',s}$ inherits this property and is also right-continuous at $s'$. Consequently, there exists a $\delta\in(0,s-s')$ such that $s'+\delta\in\mathbb{R}^+\setminus J(\omega)$ and $MD_{s'+\delta,s}>a$. \citet[Theorem 4.1]{zhang2023drawdown} implies that $MD_{s'+\delta,s}^{(n)}\to MD_{s'+\delta,s}$ as $n\to\infty$, in conjunction with the previous result yields $MD_{s'+\delta,s}^{(n)}>a$ for a sufficiently large $n$. Since $s'^{(n)}\to s'$, for a sufficiently large $n$, it follows that $s'^{(n)}\leq s'+\delta$ and hence $MD_{s'^{(n)},s}^{(n)}\geq MD_{s'+\delta,s}^{(n)}>a$,  equivalently, $\tilde{\tau}_{a,2}(\omega^{(n)})<s$. The arbitrariness of $s$ and dense property of the set $\mathbb{R}^+\setminus J(\omega)$ finally result in $\lim\sup_{n\to\infty}\tilde{\tau}_{a,2}(\omega^{(n)})\leq \tilde{\tau}_{a,2}(\omega)$.

        ~~~~Suppose $s'<s<\tilde{\tau}_{a,2}(\omega)$, $\omega\in \widetilde{\mathcal W}$, and $s\in\mathbb{R}^+\setminus J(\omega)$, then $MD_{s',s}<a$. Moreover, there exists a $\delta\in(0,s')$ such that $MD_{s'-\delta,s}<a$. The convergence $s'^{(n)}\to s'$ implies that $s'^{(n)}>s'-\delta$ holds for a sufficiently large $n$. Furthermore, applying \citet[Theorem 4.1]{zhang2023drawdown} yields $\lim_{n\to\infty}MD_{s'^{(n)},s}^{(n)}\leq\lim_{n\to\infty}MD_{s'-\delta,s}^{(n)}=MD_{s'-\delta,s}<a$, equivalently, $\tilde{\tau}_{a,2}(\omega^{(n)})>s$. As $s$ is arbitrary and $\mathbb{R}^+\setminus J(\omega)$ is dense, it follows that $\lim\inf_{n\to\infty}\tilde{\tau}_{a,2}(\omega^{(n)})\geq \tilde{\tau}_{a,2}(\omega)$.

        ~~~~Combining the above two results concludes the proof.
        \item Following the above proof of conclusion (4) in Theorem \ref{thm:convergence}, it suffices to prove the continuity of the mapping of $\omega\mapsto\tau_{a,2}(\omega)$.
        Observe that $\tau_{a,2}(\omega)=T_{M_1}+\inf\{t>T_{M_1}:\overline{\omega}_{T_{M_1},t}-X_t\geq a\}$ and $\tau_{a,2}(\omega^{(n)})=T_{M_1}^{(n)}+\inf\{t>T_{M_1}^{(n)}:\overline{\omega}_{T_{M_1}^{(n)},t}^{(n)}-Y_t^{(n)}\geq a\}$. It remains to show the convergence of $T_{M_1}^{(n)}$ as the convergence of the second term involved in $\tau_{a,2}(\omega^{(n)})$ can be proved in a similar way as that of $\tilde{\tau}_{a,2}(\omega^{(n)})$.
    
        ~~~~Suppose $s'<T_{M_1}<s,~\omega\in {V}$, and $s\in\mathbb{R}^+\setminus J(\omega)$. Thus, $\overline{\omega}_{0,s'}<\overline{\omega}_{s',s}$.
        By the
        right-continuity of the path $\omega$ at $s'$, there exists a $\delta\in(0,s-s')$ such that $s'+\delta\in\mathbb{R}^+\setminus J(\omega)$ and $\overline{\omega}_{s',s'+\delta}<\overline{\omega}_{s'+\delta,s}$. On the other hand, applications of \citet[Lemma 6]{Song2013weakconvergence} and \citet[Chapter VI, Section 2, 2.3]{jacod2013limit} yield $\overline{\omega}_{0,s'+\delta}^{(n)}\to \overline{\omega}_{0,s'+\delta}$ and $\overline{\omega}_{s'+\delta,s}^{(n)}\to \overline{\omega}_{s'+\delta,s}$ as $n\to\infty$. Besides, since $s'^{(n)}\to s'$, for a sufficiently large $n$, it follows that $s'^{(n)}<s'+\delta$. As a result, $\lim_{n\to\infty}\overline{\omega}_{0,s'^{(n)}}^{(n)}\leq \lim_{n\to\infty}\overline{\omega}_{0,s'+\delta}^{(n)}= \overline{\omega}_{0,s'+\delta}<\overline{\omega}_{s'+\delta,s}=\lim_{n\to\infty} \overline{\omega}_{s'+\delta,s}^{(n)}\leq\lim_{n\to\infty}\overline{\omega}_{s'^{(n)},s}^{(n)}$, equivalently, $s'^{(n)}<T_{M_1}^{(n)}<s$. Finally, by the arbitrariness of $s$ and dense
        property of the set $\mathbb{R}^+\setminus J(\omega)$, the inequality $\lim\sup_{n\to\infty}T_{M_1}^{(n)}\leq T_{M_1}$ holds.

        ~~~~Suppose $s'<s<T_{M_1}$, $\omega\in \mathcal W$, and $s\in\mathbb{R}^+\setminus J(\omega)$, then $\overline{\omega}_{0,s'}>\overline{\omega}_{s',s}$. Moreover, there exists a $\delta>0$ such that $s'-\delta\in\mathbb{R}^+\setminus J(\omega)$ and $\overline{\omega}_{0,s'-\delta}>\overline{\omega}_{s'-\delta,s}$. On the other hand, applying \citet[Lemma 6]{Song2013weakconvergence} and \citet[Chapter VI, Section 2, 2.3]{jacod2013limit} yield $\lim_{n\to\infty}\overline{\omega}_{0,s'-\delta}^{(n)}=\overline{\omega}_{0,s'-\delta}$ and $\lim_{n\to\infty}\overline{\omega}_{s'-\delta,s}^{(n)}=\overline{\omega}_{s'-\delta,s}$. Besides, the convergence $s'^{(n)}\to s'$ implies that $s'^{(n)}>s'-\delta$ for a sufficiently large $n$. As a result, $\lim_{n\to\infty}\overline{\omega}_{0,s'^{(n)}}^{(n)}\geq\lim_{n\to\infty}\overline{\omega}_{0,s'-\delta}^{(n)}=\overline{\omega}_{0,s'-\delta}>\overline{\omega}_{s'-\delta,s}=\lim_{n\to\infty}\overline{\omega}_{s'-\delta,s}^{(n)}\geq\lim_{n\to\infty}\overline{\omega}_{s'^{(n)},s}^{(n)}$, equivalently, $T_{M_1}^{(n)}>s$. The arbitrariness of $s$ and dense property of the set $\mathbb{R}^+\setminus J(\omega)$ finally lead to $\lim\inf_{n\to\infty}T_{M_1}^{(n)}\geq T_{M_1}$.

        ~~~~Combining the above two results completes the proof.
    \end{enumerate}
\end{proof}

\bibliographystyle{chicagoa}
\bibliography{refs}

\end{document}